%% file: BCGJLMdV24.tex
\documentclass[runningheads]{llncs}
\usepackage{color}
\usepackage{amsmath}
\usepackage{amssymb} 
\usepackage{wasysym}
\usepackage{graphicx}
\usepackage{xspace}
\usepackage{subfig}
\usepackage[export]{adjustbox}
\usepackage{stmaryrd}
\usepackage{enumitem}
\usepackage{tikz}
\usepackage{tikz-layers}
\usetikzlibrary{patterns}
\usetikzlibrary{decorations.markings}
\usepackage{xspace}
\usepackage{pgfplotstable}
\usepackage{booktabs}
\usepackage{latexsym}
\usepackage{listings}

\usepackage[T1]{fontenc}
\usetikzlibrary{automata, positioning, arrows}

\usetikzlibrary{decorations.pathmorphing}
\usetikzlibrary{calc}

\input{macros}

\begin{document}
\title{Weak Simplicial Bisimilarity for\\Polyhedral Models and \slcs$_\eta$\thanks{This is an extended version of the paper ``Weak Simplicial Bisimilarity for Polyhedral Models and \slcsE'', accepted for publication in the proceedings of the 44th International Conference on Formal Techniques for Distributed Objects, Components, and Systems (FORTE 2024) published as LNCS by Springer. It contains all detailed proofs that are not present in the FORTE 2024 paper due to lack of space.\\The authors are listed in alphabetical order, as they equally contributed to the work presented in this paper.}\\
--- Extended version ---}

\titlerunning{Weak Simplicial Bisimilarity and \slcsE{} --- Extended Version}

\author{Nick Bezhanishvili\inst{1} \and
Vincenzo Ciancia\inst{2} \and
David Gabelaia\inst{3} \and
Mamuka Jibladze\inst{3} \and
Diego Latella\inst{2} \and
Mieke Massink\inst{2} \and
Erik P. de Vink\inst{4}
}%
\authorrunning{N. Bezhanishvili et al.}
\institute{ILLC, University of Amsterdam, NL
\email{n.bezhanishvili@uva.nl}\\
\and
ISTI, Consiglio Nazionale delle Ricerche, Pisa, IT\\
\email{\{Vincenzo.Ciancia, Diego.Latella, Mieke.Massink\}@cnr.it}\\
\and
ARM Institute, Tbilisi State University, GE
\email{\{gabelaia, mamuka.jibladze\}@gmail.com}\\
\and
Eindhoven University of Technology, NL
\email{evink@win.tue.nl}
}

\maketitle              
\begin{abstract}
In the context of spatial logics and spatial model checking for polyhedral models --- mathematical basis for visualisations in continuous space --- we propose a weakening of simplicial bisimilarity. We additionally propose a corresponding weak notion of \plm-bisimilarity on cell-poset models, discrete representation of polyhedral models. We show that two points are weakly simplicial bisimilar iff their representations are weakly \plm-bisimilar. The advantage of this weaker notion is that it leads to a stronger reduction of models than its counterpart that was introduced in our previous work. This is important, since real-world polyhedral models, such as those found in domains exploiting mesh processing, typically consist of large numbers of cells. We also propose \slcsE, a weaker version of the {\em Spatial Logic for Closure Spaces} (\slcs) on polyhedral models, and we show that the proposed bisimilarities enjoy the Hennessy-Milner property: two points are weakly simplicial bisimilar iff they are logically equivalent for \slcsE. Similarly, two cells are weakly \plm-bisimilar iff they are
logically equivalent in the poset-model interpretation of \slcsE. This work is performed in the context of the geometric spatial model checker \polylogica{} and the polyhedral semantics of \slcs.

\end{abstract}
\input{Introduction}

\input{Preliminaries}

\input{SlcsE}

\input{WeakBis}

\input{example}

\input{conclusions}

\begin{credits}
\subsubsection{\ackname}
Research partially supported by
Bilateral project between CNR (Italy) and SRNSFG (Georgia) ``Model Checking for Polyhedral Logic'' (\#CNR-22-010);
European Union - Next GenerationEU - National Recovery and Resilience Plan (NRRP), Investment 1.5 Ecosystems of Innovation, Project “Tuscany Health Ecosystem” (THE), CUP: B83C22003930001;
European Union - Next-GenerationEU - National Recovery and Resilience Plan (NRRP) – MISSION 4 COMPONENT 2, INVESTIMENT N. 1.1, CALL PRIN 2022 D.D. 104 02-02-2022 – (Stendhal) CUP N. B53D23012850006;
MUR project
PRIN 2020TL3X8X ``T-LADIES'';
CNR project "Formal Methods in Software Engineering 2.0", CUP B53C24000720005.

\subsubsection{\discintname}
The authors have no competing interests to declare that are
relevant to the content of this article.

\end{credits}
%
%
%
\bibliographystyle{splncs04}
\bibliography{BCGJLMdV24}
%
\appendix

\input{Appendix}

\end{document}

%% file: macros.tex
\makeatletter
\DeclareRobustCommand{\cev}[1]{%
  \mathpalette\do@cev{#1}%
}
\newcommand{\do@cev}[2]{%
  \fix@cev{#1}{+}%
  \reflectbox{$\m@th#1\vec{\reflectbox{$\fix@cev{#1}{-}\m@th#1#2\fix@cev{#1}{+}$}}$}%
  \fix@cev{#1}{-}%
}
\newcommand{\fix@cev}[2]{%
  \ifx#1\displaystyle
    \mkern#23mu
  \else
    \ifx#1\textstyle
      \mkern#23mu
    \else
      \ifx#1\scriptstyle
        \mkern#22mu
      \else
        \mkern#22mu
      \fi
    \fi
  \fi
}

\makeatother

\newcommand{\cm}{CM}
\newcommand{\cmc}{CMC}

\newcommand{\copa}{CoPa}

\newcommand{\slcs}{{\tt SLCS}}
\newcommand{\slcsG}{{\tt SLCS}$_{\gamma}$}
\newcommand{\slcsE}{{\tt SLCS}$_{\eta}$}

\newcommand{\polylogica}{{\tt PolyLogicA}}

\newcommand{\SET}[1]{\{#1\}}
\newcommand{\ZET}[2]{\SET{#1 \,|\, #2}}
\newcommand{\pws}[1]{\mathbf{2}^{#1}}
\newcommand{\nats}{\mathbb{N}}
\newcommand{\reals}{\mathbb{R}}

\newcommand{\cnv}[1]{#1^{-}}
\newcommand{\dircnv}[1]{#1^{\pm}}

\newcommand{\plm}{$\pm$}
\newcommand{\upd}{$\uparrow\!\downarrow$}
\newcommand{\dwn}{$\downarrow$}

\newcommand{\relint}[1]{\widetilde{#1}}
\newcommand{\ap}{{\tt PL}}

\newcommand{\preceqsucc}{\preceq^{\pm}}

\newcommand{\map}{\mathbb{F}}
\newcommand{\peval}{\calV}
\newcommand{\invpeval}{\peval^{-1}}

\newcommand{\sibis}{\sim_{\resizebox{0.2cm}{!}{$\triangle$}}}
\newcommand{\wsibis}{\approx_{\resizebox{0.2cm}{!}{$\triangle$}}}

\newcommand{\plmbis}{\sim_{\pm}}
\newcommand{\wplmbis}{\approx_{\pm}}

\newcommand{\form}{\Phi}
\newcommand{\ltrue}{{\tt true}}

\newcommand{\lneg}{\neg}

\newcommand{\slcsGeq}{\equiv_{\gamma}}
\newcommand{\slcsEeq}{\equiv_{\eta}}

\newcommand{\etga}{\calE}

\newcommand{\closure}{\calC}

\newcommand{\lts}{LTS}

\newcommand{\calC}{\mathcal{C}}

\newcommand{\calE}{\mathcal{E}}
\newcommand{\calF}{\mathcal{F}}

\newcommand{\calL}{\mathcal{L}}
\newcommand{\calM}{\mathcal{M}}

\newcommand{\calP}{\mathcal{P}}

\newcommand{\calV}{\mathcal{V}}

\newcommand{\CYAN}[1]{\textcolor{cyan}{#1}}
\newcommand{\RED}[1]{\textcolor{red}{#1}}

\newcommand{\closedefi}{\hfill$\bullet$}
\newcommand{\closeex}{\hfill$\diamond$}
\newcommand{\closerem}{\hfill$\divideontimes$}
\newcommand{\sep}{\vert}

\newcounter{dgnot} 
\newenvironment{dgnot}[1][]{\refstepcounter{dgnot}\par\medskip
   \noindent \textbf{\RED{NfDiego~\thedgnot.}  #1} \rmfamily}{\medskip}

\newcounter{mknot} 
\newenvironment{mknot}[1][]{\refstepcounter{mknot}\par\medskip
   \noindent \textbf{\CYAN{NfMieke~\themknot.}  #1} \rmfamily}{\medskip}

%% file: Introduction.tex
\section{Introduction and Related Work}\label{sec:Introduction}

The notion of bisimulation is central in the theory of models for (concurrent) system behaviour, for characterising those systems which ``behave the same''. 
Properties of such behaviours are typically captured by formulas of appropriate logics, such as modal logics and variations/extensions thereof (e.g. temporal, deterministic-/stochastic-time temporal, probabilistic).
A key notion, in this context, is the Hennessy-Milner property (HMP), that allows for logical characterisations of bisimilarity.
Given a model $\calM$, a bisimulation equivalence $E$ over~$\calM$, and a logic $\calL$ intepreted on $\calM$,
we say that $E$ and $\calL$ enjoy the HMP if the following holds: any two states in $\calM$ are equivalent according to $E$ 
iff they satisfy the same formulas of $\calL$. Besides the intrinsic theoretical value of the HMP, the latter is also of fundamental importance as mathematical foundation for safe model reduction procedures since it 
ensures that any formula $\calL$ is satisfied by a state $s$ in $\calM$ iff it is satisfied by the equivalence class $[s]_E$ of $s$, that is itself a state in the {\em minimal} model $\calM_{/E}$ --- there are nowadays standard procedures for the effective and efficient computation of $\calM_{/E}$, for finite models $\calM$. Model reduction is, in turn, extremely important for efficient model analysis via, for example, automatic model-checking of logic formulas. 

Model-checking techniques have been  developed for the analysis of models of {\em space} as well, with properties expressed in {\em spatial logics}, i.e. modal logics interpreted over such models, following the tradition
going back to McKinsey and Tarski in 1940s~\cite{McKT44} (see also~\cite{vBB07} for an overview), where {\em topological models}  are considered as models for space. 
In~\cite{Ci+14,Ci+16} the
{\em Spatial Logic for Closure Spaces} (\slcs{)} has been proposed together with a model-checking algorithm for finite 
spaces and its implementation. Closure spaces are a generalisation of topological spaces that allow for a uniform treatment of  continuous spaces and discrete spaces, such as general graphs. \slcs{} is interpreted on models whose carriers are closure spaces. Spatio-temporal versions of the logic and the model-checker have been presented in~\cite{Ci+15,Ci+16a}. A version of the model-checker optimised for (2D and 3D) digital images --- that can be seen as {\em adjacency spaces}, a subclass of closure spaces --- has been proposed in~\cite{Be+19} with \slcs{} enriched with a {\em  distance} operator.  Tools for spatial and spatio-temporal model-checking have been successfully used in several applications~\cite{CLMP15,Ci+18,Ne+18,Ci+19b,Ba+20,Be+21} showing that the notions and techniques developed in the area of concurrency theory and formal methods can be extremely helpful in developing a foundational basis for the automatic  analysis of spatial models in real applications.

As a natural step forward, and following classical developments in modal logic, in~\cite{Ci+22,Ci+23}
several  notions of bisimulation for finite closure spaces have been studied. These cover a spectrum from \cm-bisimilarity, an equivalence based on {\em proximity} --- similar to and inspired by topo-bisimilarity for topological models~\cite{vBB07} --- to its specialisation for quasi-discrete closure models, \cmc-bisimilarity,  to \copa-bisimilarity, an equivalence based on {\em conditional reachability}. Each of these bisimilarities has been equipped with its logical characterisation. In~\cite{Ci+23a} an encoding from finite closure models to finite labelled transition systems (\lts{s}) has been defined and proven correct in the sense that two points in the space are \copa-bisimilar iff their corresponding states in the \lts{}
are branching bisimilar. This makes it possible to perform minimisation of the spatial model w.r.t. \copa-bisimilarity via minimisation of its \lts{} w.r.t. branching bisimulation. Very efficient tools are available for \lts{} minimisation w.r.t. branching equivalence~\cite{Gr+17}. 

The spatial model-checking techniques mentioned above have been extended to {\em polyhedral models}~\cite{Be+22,LoQ23}, that we address in the present paper.
{\em Polyhedra} are sets of points in $\reals^n$ generated by {\em simplicial complexes}, i.e. certain finite collections of {\em simplexes}, where a simplex is the {\em convex hull} of a set of affinely independent points in $\reals^n$.
Given a set $\ap$ of {\em proposition letters}, a {\em polyhedral model} is obtained from a polyhedron in the usual way, i.e. by assigning  a set of points to each proposition letter $p\in \ap$, namely those that ``satisfy'' $p$.
Polyhedral models in $\reals^3$  can be used for (approximately) representing objects in continuous 3D space. This is widely used in many 3D visual computing techniques, where an  object is divided into suitable areas of different size. Such ways of division of an object are known as {\em mesh techniques} and include triangular surface meshes or tetrahedral volume meshes (see for example~\cite{LevinePRZ2012} and the example in Fig.~\ref{subfig:cube}).

In~\cite{Be+22} a version of \slcs{,} referred to as  \slcsG{} in this paper, has been proposed for expressing spatial properties of points laying in polyhedral models, and in particular {\em conditional reachability} properties. Many other interesting properties, such as ``being surrounded by'' can be expressed using reachability (see~\cite{Be+22}). 
Intuitively, a point $x$ in a polyhedral model satisfies the conditional reachability formula
$\gamma(\form_1,\form_2)$ if there is a topological path starting from $x$, ending in a point $y$ satisfying $\form_2$,
and such that all the intermediate points of the path between $x$ and $y$ satisfy $\form_1$, but note that neither $x$ nor $y$ is required to satisfy $\form_1$. 
A notion of bisimilarity between points has also been introduced in~\cite{Be+22}, namely {\em simplicial bisimilarity}, and it has been proven that the latter enjoys the HMP w.r.t. \slcsG{.} In addition, a representation $\map$ of polyhedral models as finite posets has been built and it has been shown that a point $x$ in a polyhedral model $\calP$ satisfies
a \slcsG{} formula $\form$ in $\calP$ iff its representation $\map(x)$ in the poset model $\map(\calP)$ representing $\calP$ satisfies $\form$ in $\map(\calP)$. An  \slcsG{} model-checking algorithm has been developed for finite poset models that has been implemented in the tool \polylogica{} thus achieving model-checking of continuous space
that can be represented by polyhedral models (see~\cite{Be+22} for details). 

In~\cite{Ci+23c} we addressed the issue of minimisation for polyhedral models and, more specifically, their poset representations.
In particular, we defined \plm\nobreakdash-bisimilarity, a notion based on \plm-paths, a subclass of undirected paths over poset models suitable for representing, in such models, topological paths over polyhedral models. We proved that
\plm-bisimilarity enjoys the HMP w.r.t. \slcsG{} interpreted on finite poset models and we showed that it can be used
for poset model minimisation: for instance, the minimal model of the poset model of Fig.~\ref{subfig:PolyhedronNoPathPosetCompressed} has only 10 elements. 

In this paper we present \slcsE{,} another variant of \slcs{} for polyhedral models where the $\gamma$ modality has been replaced by $\eta$ so that a point $x$ satisfies
$\eta(\form_1,\form_2)$ if there is a topological path starting from $x$, ending in a point $y$ satisfying $\form_2$,
and such that  all the intermediate points of the path between $x$ and $y$,  {\em and including $x$ itself},  satisfy $\form_1$ ($y$ is not required to satisfy $\form_1$). Thus $\gamma$ and $\eta$ behave differently {\em only} in $\eta$ requiring that $x$ itself satisfies $\form_1$.

The result is that \slcsE{} is weaker than \slcsG{} in the sense that it distinguishes less points than \slcsG{.} Furthermore, \slcsG{} can express proximity, here intended as topological closure --- that boils down to the  standard {\em possibly} modality $\Diamond$ in the poset model interpretation --- whereas \slcsE{} cannot. Nevertheless, many interesting reachability properties can be expressed in \slcsE{} and, perhaps most importantly, the latter characterises bisimilarities (in the polyhedral model and the associated poset model) that are coarser than simplicial bisimilarity and \plm-bisimilarity, respectively.
This allows for a substantial model reduction. For instance, the minimal model, w.r.t. the new equivalence, of the poset of Fig.~\ref{subfig:PolyhedronNoPathPosetCompressed},
shown in Fig.~\ref{fig:exa:MinRunExaE}, has only 4 states now. This greater reduction in model size is one of the main motivations for the study of \slcsE{} presented in this paper. 

In the remainder of this paper, we provide necessary background information in Sect.~\ref{sec:BackAndNotat}. 
Sect.~\ref{sec:slcsE} introduces \slcsE{} and addresses its relationship with \slcsG{.} It is also shown that \slcsE{} is
preserved and reflected by the mapping $\map$ form  polyhedral models to finite poset models. 
Weak simplicial bisimilarity and weak \plm-bisimilarity are defined in
Sect.~\ref{sec:WeakBis} where it is also shown that they enjoy the HMP w.r.t. the intepretation of \slcsE{}
on polyhedral models and on finite poset models, respectively. 
A larger example is given in Sect.~\ref{sec:example}, illustrating the reduction potential of weak \plm-bisimulation.
Finally, conclusions and a discussion on
future work are reported in Sect.~\ref{sec:ConclusionsFW}.
The proofs for all the results presented in Sections~\ref{sec:BackAndNotat}, \ref{sec:slcsE} and \ref{sec:WeakBis} are provided in Appendix~\ref{apx:DetailedProofs} whereas in Appendix~\ref{apx:BackgroundInDetail}
we recall background information and results, as well as notational details.

%% file: Preliminaries.tex
\section{Background and Notation}\label{sec:BackAndNotat}

In this section we recall the relevant details of the language \slcsG{,} its polyhedral and poset models, and the truth-preserving map $\map$ between these models.

For sets $X$ and $Y$, a function $f:X \to Y$, and subsets $A \subseteq X$ and $B \subseteq Y$ 
we define $f(A)$ and $f^{-1}(B)$ as $\ZET{f(a)}{a \in A}$  and $\ZET{a}{f(a) \in B}$, respectively. 
The  {\em restriction} of  $f$ on $A$ is denoted by $f|A$.
The powerset of $X$ is denoted by $\pws{X}$.
For a relation $R\subseteq X\times X$ we let $\cnv{R}=\ZET{(y,x)}{(x,y)\in R}$ denote its converse and $\dircnv{R}$ denote $R \, \cup  \cnv{R}$.
In the remainder of the paper we assume that a set  $\ap$ of {\em proposition letters} is fixed.
The sets of natural numbers and of real numbers are denoted by $\nats$ and $\reals$, respectively. 
We use the standard interval notation: for $x,y \in \reals$ we let $[x,y]$ be the set
$\ZET{r\in \reals}{x\leq r \leq y}$, $[x,y) = \ZET{r\in \reals}{x\leq r < y}$, and so on.
Intervals of $\reals$ are equipped with the Euclidean topology inherited from $\reals$.
We use a similar notation for intervals over $\nats$:  
for $n,m \in \nats$ 
$[m;n]$ denotes the set $\ZET{i\in\nats}{m\leq i \leq n}$,
$[m;n)=\ZET{i\in\nats}{m\leq i < n}$, and so on.

Below we recall, informally, some basic notions, assuming that the reader is familiar with topological spaces, Kripke models and posets. For all the details concerning basic notions and notation we refer the reader to~\cite{Be+22,Ci+23c}.

A {\em simplex} $\sigma$ is the convex hull of a set of  $d+1$ affinely independent points in $\reals^m$, with $d \leq m$, 
i.e. $\sigma = \ZET{ \lambda_0\mathbf{v_0} + \ldots + \lambda_d\mathbf{v_d}}{\lambda_0,\ldots,\lambda_d \in [0,1]\mbox{ and }
\sum_{i=0}^{d} \lambda_i = 1}$. For instance, a segment $AB$  together with its end-points $A$ and $B$ 
is a simplex in $\reals^m$, for $m\geq 1$. 
Given a simplex $\sigma$ with  vertices $\mathbf{v_0},\ldots,\mathbf{v_d}$, any
subset of $\SET{\mathbf{v_0},\ldots,\mathbf{v_d}}$  spans a simplex $\sigma'$ in turn: we say that $\sigma'$ is a {\em face} of $\sigma$,
written $\sigma' \sqsubseteq \sigma$. 
So, for instance, both $A$ and $B$ are simplexes in turn --- with $A \sqsubseteq AB$ and
$B \sqsubseteq AB$ --- and $AB$ itself could be part of a larger simplex, e.g., a triangle $ABC$.
Clearly, $\sqsubseteq$ is a partial order. 
The {\em barycentre} $b_{\sigma}$ of $\sigma$ is defined as  follows:
$
b_{\sigma} =
\sum_{i=0}^d \frac{1}{d+1}\mathbf{v_i}
$.

The relative interior $\relint{\sigma}$ of a simplex $\sigma$ is the same as $\sigma$ ``without its borders'', i.e. the set $\ZET{ \lambda_0\mathbf{v_0} + \ldots + \lambda_d\mathbf{v_d}}{\lambda_0,\ldots,\lambda_d \in (0,1]\mbox{ and } \sum_{i=0}^{d} \lambda_i = 1}$. For instance, the open segment $\relint{AB}$, without the
end-points $A$ and $B$ is the relative interior of segment $AB$. The relative interior of a simplex is often called {\em a cell} and is equal to the topological interior taken inside the affine hull of the simplex.\footnote{But note that the relative interior of a simplex composed of just a single point is the point itself and not the empty set.} There is an obvious partial order between the cells of a simplex: 
$\relint{\sigma_1} \preceq \relint{\sigma_2}$ iff $\relint{\sigma_1} \subseteq \closure_T(\relint{\sigma_2})$,
where $\closure_T$ denotes the classical topological closure operator. So, 
in the above example, we have $\relint{A}\preceq \relint{A}, \relint{B}\preceq \relint{B}, \relint{A}\preceq \relint{AB}$, $\relint{B}\preceq \relint{AB}$, and $\relint{AB}\preceq \relint{AB}$.
Note that for all simplexes $\sigma_1$ and $\sigma_2$ the following holds: $\sigma_1 \sqsubseteq \sigma_2$ iff 
$\relint{\sigma_1} \preceq \relint{\sigma_2}$.

A {\em simplicial complex} $K$
 is a finite collection of simplexes of $\reals^m$ such that:
(i) if $\sigma \in K$ and $\relint{\sigma}' \preceq \relint{\sigma}$ then also $\sigma' \in K$; 
(ii) if $\sigma, \sigma' \in K$ then 
$\relint{\sigma \cap \sigma'} \preceq \relint{\sigma}\cap \relint{\sigma}'$.
Given a simplicial complex $K$,  the {\em cell poset} of $K$ is the poset $(\relint{K},\preceq)$ where
$\relint{K}$ is the set $\ZET{\relint{\sigma}}{\sigma \in K\setminus \SET{\emptyset}}$ and
the {\em polyhedron} $|K|$ of $K$ is the set-theoretic union of the simplexes in $K$. Note that
$|K|$ inherits the topological structure of $\reals^m$. 

A {\em polyhedral model} is a pair $(|K|,V)$
where $V: \ap \to \pws{|K|}$ maps every proposition letter $p\in \ap$ to the set of points of $|K|$ satisfying $p$. It is required that, for all $p\in \ap$, $V(p)$ is always a union of cells in $\relint{K}$.
Similarly, a poset model $(W,\preceq,\peval)$ is a poset equipped with a  valuation
function $\peval:\ap \to \pws{W}\!$. Given polyhedral model $\calP= (|K|,V)$, we say that
 $(\relint{K},\preceq,\peval)$ is the  {\em cell poset model} of $\calP$  iff 
 $(\relint{K},\preceq)$ is the cell poset of $K$ and, for all $\relint{\sigma}\in \relint{K}$, we have:
 $\relint{\sigma} \in \peval(p)$ iff  $\relint{\sigma} \subseteq V(p)$. 
 We let  $\map(\calP)$ denote the cell poset model of $\calP$ and, with a little bit of overloading, for all $x\in |K|$, $\map(x)$ denotes the
 unique cell $\relint{\sigma}$ such that $x \in \relint{\sigma}$. Note that $\map: |K| \to \relint{K}$ is a continuous function~\cite[Corollary 3.4]{BMMP2018}. Furthermore, note that poset models are a subclass of Kripke models. In the sequel, when we say that $\calF$ is a cell poset model, we mean that there exists a polyhedral model $\calP$ such that $\calF=\map(\calP)$.

Fig.~\ref{fig:PolyhedronNoPathCompressed} shows a polyhedral model. There are three 
proposition letters, $\mathbf{red}$, $\mathbf{green}$ and $\mathbf{gray}$, shown by different colours (\ref{subfig:PolyhedronNoPathCompressed}). The model is ``unpacked'' into its cells in Fig.~\ref{subfig:PolyhedronNoPathCellsCompressed}. The latter are collected in the cell poset model, whose Hasse diagram is shown in Fig.~\ref{subfig:PolyhedronNoPathPosetCompressed}.

\begin{figure}[h]
\subfloat[]{\label{subfig:PolyhedronNoPathCompressed}
\resizebox{0.9in}{!}
{
\begin{tikzpicture}[scale=1.4,label distance=-2pt]
	    \tikzstyle{point}=[circle,draw=black,fill=white,inner sep=0pt,minimum width=4pt,minimum height=4pt]
	    \node (p0)[point,draw=red,label={270:$B$}] at (0,0) {};
	    	\filldraw [red] (p0) circle (1.25pt);
	    \node (p1)[point,draw=gray,label={ 90:$A$}] at (0,1) {};
	    	\filldraw [gray] (p1) circle (1.25pt);
	    \node (p2)[point,draw=gray,label={270:$D$}] at (1,0) {};
	    \node (p3)[point,draw=red,label={ 90:$C$}] at (1,1) {};
	    	\filldraw [red] (p3) circle (1.25pt);
	    \node (p4)[point,draw=gray,label={270:$F$}] at (2,0) {};
	    \node (p5)[point,draw=gray,label={ 90:$E$}] at (2,1) {};

	    \draw [red   ,thick](p0) -- (p1);
	    \draw [red   ,thick](p0) -- (p2);
	    \draw [red   ,thick](p0) -- (p3);
	    \draw [red   ,thick](p1) -- (p3);
	    \draw [red   ,thick](p2) -- (p3);	    
    \draw [dashed      ](p2) -- (p4);
    \draw [dashed      ](p2) -- (p5);
    \draw [dashed      ](p3) -- (p5);
    \draw [dashed      ](p4) -- (p5);
    \draw [gray,thick](p2) -- (p4);
    \draw [gray,thick](p4) -- (p5);
    \draw [gray,thick](p2) -- (p5);
    \draw [gray,thick](p3) -- (p5);
	        
	    \begin{scope}[on background layer]
	    \fill [fill=red!50  ](p0.center) -- (p1.center) -- (p3.center);
	    \fill [fill=red!50  ](p0.center) -- (p3.center) -- (p2.center);
	    \fill [fill=green!50](p2.center) -- (p3.center) -- (p5.center);
            \fill [fill=gray!50](p2.center) -- (p4.center) -- (p5.center);    	    
            \end{scope}

    \filldraw [gray] (p2) circle (1.25pt);
    \filldraw [gray] (p4) circle (1.25pt);
    \filldraw [gray] (p5) circle (1.25pt); 
	\end{tikzpicture}
	}
}
\subfloat[]{\label{subfig:PolyhedronNoPathCellsCompressed}
\resizebox{1.5in}{!}
{
\begin{tikzpicture}[scale=1.3,label distance=-2pt]
	    \tikzstyle{point}=[circle,fill=white,inner sep=0pt,minimum width=4pt,minimum height=4pt]
	    \node (p0S0d)[point,draw=red,fill=red,label={270:$B$}] at (0,0) {};
	    \node (p0S1d)[point] at (0.33,0) {};
	    \node (p0S2d)[point] at (0.66,0) {};
	    \node (p0S3d)[point] at (0.99,0) {};
	    \node (p0S0u)[point] at (0,0.33) {};
	    \node (p0S1u)[point] at (0.33,0.33) {};
	    \node (p0S2u)[point] at (0.66,0.33) {};
	    \node (p0S3u)[point] at (0.99,0.33) {};
	    
	    \node (p1S0d)[point] at (0,1.33) {};
	    \node (p1S1d)[point] at (0.33,1.33) {};
	    \node (p1S0u)[point,fill=gray,label={90:$A$}] at (0,1.66) {};
	    \node (p1S1u)[point] at (0.33,1.66) {};
	    
	    \node (p2S0u)[point] at (1.99,0.33) {};
	    \node (p2S0d)[point] at (1.99,0) {};
	    \node (p2S1u)[point] at (2.32,0.33) {};
	    \node (p2S1d)[point,fill=gray,label={270:$D$}] at (2.32,0.0) {};
	    \node (p2S2u)[point] at (2.65,0.33) {};
	    \node (p2S3u)[point] at (2.98,0.33) {};
	    \node (p2S4u)[point] at (3.31,0.33) {};
	    \node (p2S4d)[point] at (3.31,0.0) {};
	    
	    \node (p3S0u)[point] at (1.33,1.66) {};
	    \node (p3S1u)[point,fill=red,label={90:$C$}] at (2.32,1.66) {};
	    \node (p3S0d)[point] at (1.33,1.33) {};
	    \node (p3S1d)[point] at (1.66,1.33) {};
	    \node (p3S2u)[point] at (2.65,1.66) {};
	    \node (p3S3u)[point] at (2.98,1.66) {};
	    \node (p3S2d)[point] at (1.99,1.33) {};
	    \node (p3S3d)[point] at (2.32,1.33) {};
	    \node (p3S4d)[point] at (2.65,1.33) {};
	    
	    \node (p4S0u)[point] at (4.31,0.33) {};
	    \node (p4S0d)[point] at (4.31,0.0) {};
	    \node (p4S1u)[point] at (4.64,0.33) {};
	    \node (p4S1d)[point,fill=gray,label={270:$F$}] at (4.64,0.0) {};	    
	    
	    \node (p5S0u)[point] at (3.65,1.66) {};
	    \node (p5S1u)[point,fill=gray,label={90:$E$}] at (4.64,1.66) {};
	    \node (p5S0d)[point] at (3.65,1.33) {};
	    \node (p5S1d)[point] at (3.98,1.33) {};
	    \node (p5S2d)[point] at (4.31,1.33) {};
	    \node (p5S3d)[point] at (4.64,1.33) {};
	    
	    \draw [red,thick](p0S0u) -- (p1S0d);
	    \draw [red,thick](p1S1u) -- (p3S0u);
	    \draw [red,thick](p0S2u) -- (p3S1d);
	    \draw [red,thick](p0S3d) -- (p2S0d);
	    \draw [red,thick](p3S3d) -- (p2S1u);
	    \draw [gray,thick](p2S3u) -- (p5S1d);
	    \draw [gray,thick](p3S2u) -- (p5S0u);
	    \draw [gray,thick](p2S4d) -- (p4S0d);
	    \draw [gray,thick](p5S3d) -- (p4S1u);
	    	    
	    \begin{scope}[on background layer]
	    \fill [fill=red!50  ](p0S1u.center) -- (p1S1d.center) -- (p3S0d.center);
 	    \fill [fill=red!50  ](p0S3u.center) -- (p3S2d.center) -- (p2S0u.center);
	    \fill [fill=green!50  ](p2S2u.center) -- (p3S4d.center) -- (p5S0d.center);
	    \fill [fill=gray!50  ](p2S4u.center) -- (p5S2d.center) -- (p4S0u.center);    
            \end{scope}
	\end{tikzpicture}
	}
}
\subfloat[]{\label{subfig:PolyhedronNoPathPosetCompressed}
\resizebox{2in}{!}
{
\begin{tikzpicture}[scale=20, every node/.style={transform shape}]
    \tikzstyle{kstate}=[rectangle,draw=black,fill=white]
    \tikzset{->-/.style={decoration={
		markings,
		mark=at position #1 with {\arrow{>}}},postaction={decorate}}}
    
    \node[kstate,fill=red!50  ] (P0) at (  1,0) {$\relint{B}$};
    \node[kstate,fill=lightgray!50  ] (P1) at (  0,0) {$\relint{A}$};
    \node[kstate,fill=lightgray!50] (P2) at (3.5,0) {$\relint{D}$};
    \node[kstate,fill=red!50  ] (P3) at (2.5,0) {$\relint{C}$};
    \node[kstate,fill=lightgray!50] (P4) at (  6,0) {$\relint{F}$};
    \node[kstate,fill=lightgray!50] (P5) at (  5,0) {$\relint{E}$};

    \node[kstate,fill=red!50] (E0) at (-1,1) {$\relint{AB}$};
    \node[kstate,fill=red!50] (E1) at ( 2,1) {$\relint{BD}$};
    \node[kstate,fill=red!50] (E2) at ( 1,1) {$\relint{BC}$};
    \node[kstate,fill=red!50  ] (E3) at ( 0,1) {$\relint{AC}$};
    \node[kstate,fill=red!50  ] (E4) at ( 3,1) {$\relint{CD}$};
	\node[kstate,fill=lightgray!50] (E5) at ( 6,1) {$\relint{DF}$};
	\node[kstate,fill=lightgray!50] (E6) at ( 5,1) {$\relint{DE}$};
	\node[kstate,fill=lightgray!50] (E7) at ( 4,1) {$\relint{CE}$};
	\node[kstate,fill=lightgray!50] (E8) at ( 7,1) {$\relint{EF}$};

    \node[kstate,fill=red!50] (T0) at ( 2,2) {$\relint{BCD}$};
    \node[kstate,fill=red!50] (T1) at ( 0,2) {$\relint{ABC}$};
    \node[kstate,fill=lightgray!50] (T2) at ( 6,2) {$\relint{DEF}$};
    \node[kstate,fill=green!50] (T3) at ( 4,2) {$\relint{CDE}$};

    \draw (P0) to (E0);
    \draw (P0) to (E1);
    \draw (P0) to (E2);

    \draw (P1) to (E0);
    \draw (P1) to (E3);

    \draw (P2) to (E1);
    \draw (P2) to (E4);
    \draw (P2) to (E5);
    \draw (P2) to (E6);

    \draw (P3) to (E2);
    \draw (P3) to (E3);
    \draw (P3) to (E4);
    \draw (P3) to (E7);

    \draw (P4) to (E5);
    \draw (P4) to (E8);

    \draw (P5) to (E6);
    \draw (P5) to (E7);
    \draw (P5) to (E8);
    \draw (E0) to (T1);
    \draw (E2) to (T0);
    \draw (T1) to (E2);


    \draw (E1) to (T0);
    
%

    \draw (E3) to (T1);
    
    \draw (E4) to (T0);
    \draw (E4) to (T3);

	\draw (E5) to (T2);

	\draw (E6) to (T2);
	\draw (E6) to (T3);

	\draw (E7) to (T3);

	\draw (E8) to (T2);


\end{tikzpicture}
}
}
\caption{A polyhedral model $\calP$ (\ref{subfig:PolyhedronNoPathCompressed}) with its cells (\ref{subfig:PolyhedronNoPathCellsCompressed}) and the Hasse diagram of the related cell poset (\ref{subfig:PolyhedronNoPathPosetCompressed}).}
\label{fig:PolyhedronNoPathCompressed}
\end{figure}
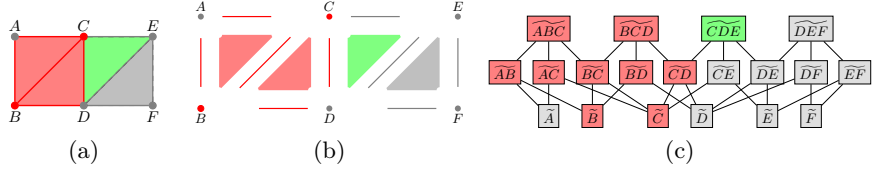

In a topological space $(X,\tau)$, a {\em topological path} from $x\in X$ is a total, continuous function $\pi : [0,1] \to X$ such that $\pi(0)=x$.
We call $\pi(0)$ and $\pi(1)$ the {\em starting}  and {\em ending}  {\em point} of $\pi$, respectively, while
 $\pi(r)$ is an {\em intermediate point} of $\pi$, for all $r\in(0,1)$. Fig.~\ref{subfig:PolyhedronWithPath} shows a path from a point $x$ in the open segment $\relint{AB}$ in the polyhedral model of Fig~\ref{subfig:PolyhedronNoPathCompressed}. 

Topological paths are represented in cell posets by so-called \plm-paths, a subclass of undirected paths~\cite{Be+22}. 
For technical reasons\footnote{We are interested in model-checking structures resulting from the minimisation, via bisimilarity, of cell poset models, and such structures are often just (reflexive) Kripke models rather than poset models.} in this paper we extend the definition given in~\cite{Be+22} to general Kripke frames. Given a Kripke frame $(W,R)$, an {\em undirected  path} of length $\ell \in \nats$ from $w$ is a total function $\pi : [0;\ell] \to X$  such that $\pi(0)=x$ and, for all $i\in [0;\ell)$,  $\dircnv{R}(\pi(i),\pi(i+1))$. The
{\em starting} and {\em ending points} are $\pi(0)$ and $\pi(\ell)$, respectively, while $\pi(i)$ is an intermediate point for all $i\in(0;\ell)$. The path is a {\em \plm-path} iff $\ell\geq 2$, $R(\pi(0),\pi(1))$ and $\cnv{R}(\pi(\ell-1),\pi(\ell))$. 

The \plm-path 
$(\relint{AB},\relint{ABC},\relint{BC},\relint{BCD},\relint{D})$\footnote{For undirected path $\pi$ of length $\ell$ we often use the sequence notation $(x_i)_{i=0}^{\ell}$ where
$x_i=\pi(i)$ for $i\in [0;\ell]$.}, drawn in blue in Fig.~\ref{subfig:PosetWithPath}, faithfully represents the path from $x$ shown in Fig.~\ref{subfig:PolyhedronWithPath}. Note that a path $\pi$ such that, say, $\pi(0)\in \relint{CD}$, 
$\pi(1) = E$ and $\pi((0,1))\subseteq \relint{CDE}$, i.e. a path that ``jumps immediately'' to  
 $\relint{CDE}$ after starting in $\relint{CD}$ cannot be represented in the poset by any undirected path $\pi'$, of some length $\ell\geq 2$ such that $\pi'(0) \succ \pi'(1)$ (or $\pi'(\ell-1) \prec \pi'(\ell)$, for symmetry reasons), while it is correctly represented by the \plm-path 
 $(\relint{CD},\relint{CDE},\relint{E})$, where $\relint{CD} \prec \relint{CDE} \succ \relint{E}$.
 
In the context of this paper it is often convenient to use a generalisation of \plm-paths,
so-called ``down paths'', \dwn-paths for short: a \dwn-path from $w$, of length $\ell \geq 1$, is an undirected path $\pi$ from $w$ of length $\ell$ such that  $\cnv{R}(\pi(\ell-1),\pi(\ell))$.
Clearly, every \plm-path is also a \dwn-path. The following lemma ensures that in {\em reflexive} Kripke frames \plm- and \dwn-paths can be safely used interchangeably since for every \dwn-path there is a \plm-path with the same starting and ending points and with the same set of intermediate points, occurring in the same order:

\begin{lemma}\label{lem:d2plm}
Given a reflexive Kripke frame  $(W,R)$ and a \dwn-path $\pi:[0;\ell]\to W$,
there is a \plm-path $\pi':[0;\ell'']\to W$, for some $\ell'$, and a 
total, surjective, monotonic, non-decreasing function $f:[0;\ell'] \to [0;\ell]$ with  
$\pi'(j)=\pi(f(j))$ for all $j\in [0;\ell']$.
\end{lemma}

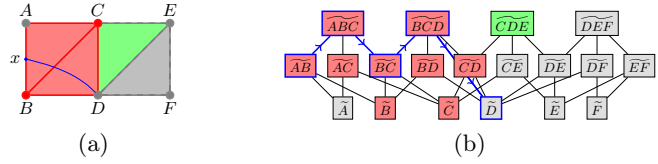
\begin{figure}[t]
\begin{center}
\subfloat[]{\label{subfig:PolyhedronWithPath}
\resizebox{1in}{!}{
\begin{tikzpicture}[scale=1.3,label distance=-2pt]
	    \tikzstyle{point}=[circle,draw=black,fill=white,inner sep=0pt,minimum width=4pt,minimum height=4pt]
	    \node (p0)[point,draw=red,label={270:$B$}] at (0,0) {};
	    	\filldraw [red] (p0) circle (1.25pt);
	    \node (p1)[point,draw=gray,label={ 90:$A$}] at (0,1) {};
	    	\filldraw [gray] (p1) circle (1.25pt);
	    \node (p2)[point,draw=gray,label={270:$D$}] at (1,0) {};
	    \node (p3)[point,draw=red,label={ 90:$C$}] at (1,1) {};
	    	\filldraw [red] (p3) circle (1.25pt);
	    \node (p4)[point,draw=gray,label={270:$F$}] at (2,0) {};
	    \node (p5)[point,draw=gray,label={ 90:$E$}] at (2,1) {};

	    \draw [red   ,thick](p0) -- (p1);
	    \draw [red   ,thick](p0) -- (p2);
	    \draw [red   ,thick](p0) -- (p3);
	    \draw [red   ,thick](p1) -- (p3);
	    \draw [red   ,thick](p2) -- (p3);	    
    \draw [dashed      ](p2) -- (p4);
    \draw [dashed      ](p2) -- (p5);
    \draw [dashed      ](p3) -- (p5);
    \draw [dashed      ](p4) -- (p5);
    \draw [gray,thick](p2) -- (p4);
    \draw [gray,thick](p4) -- (p5);
    \draw [gray,thick](p2) -- (p5);
    \draw [gray,thick](p3) -- (p5);
	        
	    \begin{scope}[on background layer]
	    \fill [fill=red!50  ](p0.center) -- (p1.center) -- (p3.center);
	    \fill [fill=red!50  ](p0.center) -- (p3.center) -- (p2.center);
	    \fill [fill=green!50](p2.center) -- (p3.center) -- (p5.center);
            \fill [fill=gray!50](p2.center) -- (p4.center) -- (p5.center);    	    
            \end{scope}

	    \node at (-.15,.5) {$x$};
	    \fill[blue] (0,.5) circle (.7pt);
	    \draw [blue] plot [smooth,tension=1] coordinates { (0,.5) (.6,.3) (1,0)};
	    \fill[blue] (1,0) circle (.7pt);
    \filldraw [gray] (p2) circle (1.25pt);
    \filldraw [gray] (p4) circle (1.25pt);
    \filldraw [gray] (p5) circle (1.25pt); 
	\end{tikzpicture}
}
	}\quad\quad\quad
\subfloat[]{\label{subfig:PosetWithPath}
\resizebox{2in}{!}{
\begin{tikzpicture}[scale=0.8, every node/.style={transform shape}]
    \tikzstyle{kstate}=[rectangle,draw=black,fill=white]
    \tikzset{->-/.style={decoration={
		markings,
		mark=at position #1 with {\arrow{>}}},postaction={decorate}}}
    
    \node[kstate,fill=red!50  ] (P0) at (  1,0) {$\relint{B}$};
    \node[kstate,fill=lightgray!50  ] (P1) at (  0,0) {$\relint{A}$};
    \node[kstate,fill=lightgray!50,draw=blue,thick] (P2) at (3.5,0) {$\relint{D}$};
    \node[kstate,fill=red!50  ] (P3) at (2.5,0) {$\relint{C}$};
    \node[kstate,fill=lightgray!50] (P4) at (  6,0) {$\relint{F}$};
    \node[kstate,fill=lightgray!50] (P5) at (  5,0) {$\relint{E}$};

    \node[kstate,fill=red!50  ,draw=blue,thick] (E0) at (-1,1) {$\relint{AB}$};
    \node[kstate,fill=red!50  ] (E1) at ( 2,1) {$\relint{BD}$};
    \node[kstate,fill=red!50  ,draw=blue,thick] (E2) at ( 1,1) {$\relint{BC}$};
    \node[kstate,fill=red!50  ] (E3) at ( 0,1) {$\relint{AC}$};
    \node[kstate,fill=red!50  ] (E4) at ( 3,1) {$\relint{CD}$};
	\node[kstate,fill=lightgray!50] (E5) at ( 6,1) {$\relint{DF}$};
	\node[kstate,fill=lightgray!50] (E6) at ( 5,1) {$\relint{DE}$};
	\node[kstate,fill=lightgray!50] (E7) at ( 4,1) {$\relint{CE}$};
	\node[kstate,fill=lightgray!50] (E8) at ( 7,1) {$\relint{EF}$};

    \node[kstate,fill=red!50  ,draw=blue,thick] (T0) at ( 2,2) {$\relint{BCD}$};
    \node[kstate,fill=red!50  ,draw=blue,thick] (T1) at ( 0,2) {$\relint{ABC}$};
    \node[kstate,fill=lightgray!50] (T2) at ( 6,2) {$\relint{DEF}$};
    \node[kstate,fill=green!50] (T3) at ( 4,2) {$\relint{CDE}$};

    \draw (P0) to (E0);
    \draw (P0) to (E1);
    \draw (P0) to (E2);

    \draw (P1) to (E0);
    \draw (P1) to (E3);

    \draw (P2) to (E1);
    \draw (P2) to (E4);
    \draw (P2) to (E5);
    \draw (P2) to (E6);

    \draw (P3) to (E2);
    \draw (P3) to (E3);
    \draw (P3) to (E4);
    \draw (P3) to (E7);

    \draw (P4) to (E5);
    \draw (P4) to (E8);

    \draw (P5) to (E6);
    \draw (P5) to (E7);
    \draw (P5) to (E8);

    \draw[blue,thick,->-=.5] (E0) to (T1);

    \draw (E1) to (T0);
    
    \draw[blue,thick,->-=.5] (E2) to (T0);
    
    \draw[blue,thick,->-=.5] (T1) to (E2);

    \draw (E3) to (T1);
    
    \draw (E4) to (T0);
    \draw (E4) to (T3);

	\draw (E5) to (T2);

	\draw (E6) to (T2);
	\draw (E6) to (T3);

	\draw (E7) to (T3);

	\draw (E8) to (T2);

	\begin{scope}[on background layer]
		\draw[blue,thick,->-=.2] (T0) to (P2);
		\draw[blue,thick,->-=.8] (T0) to (P2);
	\end{scope}
\end{tikzpicture}
}
}
\end{center}
\caption{ (\ref{subfig:PolyhedronWithPath}) A topological path from a point $x$ to vertex $D$ in the polyhedral model $\calP$ of Figure~\ref{subfig:PolyhedronNoPathCompressed}. (\ref{subfig:PosetWithPath}) The corresponding \plm-path (in blue) in the Hasse diagram of the cell poset model $\map(\calP)$.
}\label{fig:poset}
\end{figure}

In~\cite{Be+22}, \slcsG{,} a version of \slcs{} for polyhedral models, has been presented that consists of predicate letters, negation, conjunction, and the single modal operator~$\gamma$, expressing conditional reachability.
The satisfaction relation for $\gamma(\form_1, \form_2)$, for polyhedral model $\calP=(|K|,V)$ and
$x\in |K|$, as defined in~\cite{Be+22}, is recalled below:\\[0.5em]
$
\begin{array}{l c l}
\calP, x \models \gamma(\form_1,\form_2) & \Leftrightarrow &
\mbox{a } \mbox{topological path } \pi: [0,1] \to |K| \mbox{ exists such that } \pi(0)=x,\\&&
\calP, \pi(1) \models \form_2, \mbox{and }
\calP, \pi(r) \models \form_1 \mbox{ for all } r\in \!\!(0,\!1).
\end{array}
$\\[0.5em]
We also recall the interpretation of \slcsG{} on poset models. The satisfaction relation for $\gamma(\form_1, \form_2)$, for poset model $\calF=(W,\preceq,\peval)$ and
$w\in W$, is as follows:\\[0.5em]
$
\begin{array}{l c l}
\calF, w \models \gamma(\form_1,\form_2) & \Leftrightarrow &
\mbox{a } \mbox{\plm-path } \pi: [0;\ell] \to W \mbox{ exists such that } \pi(0)=w,\\&&
\calF, \pi(\ell) \models \form_2, \mbox{and }
\calF, \pi(i) \models \form_1 \mbox{ for all } i\in \!\!(0;\!\ell).
\end{array}
$\\[0.5em]
In~\cite{Be+22} it has also been shown that, for all $x\in |K|$ and \slcsG{} formulas $\form$, we have:
$\calP,x \models \form$ iff $\map(\calP),\map(x) \models \form$.
In addition, {\em simplicial bisimilarity}, a novel notion of bisimilarity for polyhedral models, has been defined that uses a subclass of topological paths and it has been shown
to enjoy the classical Hennessy-Milner property: two points $x_1,x_2 \in |K|$ are 
simplicial bisimilar, written $x_1 \sibis^{\calP} x_2$, iff they satisfy the same \slcsG{} formulas, i.e. they are equivalent with 
respect to the logic \slcsG{}, written $x_1 \slcsGeq^{\calP} x_2$.

The result has been extended to {\em \plm-bisimilarity} on finite poset models, a notion of
bisimilarity based on \plm-paths:
$w_1,w_2 \in W$ are  \plm-bisimilar, written $x_1 \plmbis^{\calF} x_2$, iff they satisfy the same \slcsG{} formulas, i.e.  $x_1 \slcsGeq^{\calF} x_2$ (see~\cite{Ci+23c} for details). 
In summary, we have:
\begin{equation}\label{HMext}
x_1 \sibis^{\calP} x_2 \mbox{ iff }
x_1 \slcsGeq^{\calP} x_2 \mbox{ iff }
\map(x_1) \slcsGeq^{\map(\calP)} \map(x_2) \mbox{ iff }
\map(x_1)\plmbis^{\map(\calP)} \map(x_2).
\end{equation}
As an illustration, with reference to Figure~\ref{subfig:PolyhedronNoPathCompressed}, we have that no 
red point, call it  $x$, in the open segment $CD$ is simplicial bisimilar to the red point $C$. In fact, although both $x$ and $C$ satisfy $\gamma(\mathbf{green}, \ltrue)$, we have that 
$C$ satisfies also $\gamma(\mathbf{gray}, \ltrue)$, which is not the case for $x$.
Similarly, with reference to Figure~\ref{subfig:PolyhedronNoPathPosetCompressed}, cell $\relint{C}$ satisfies 
$\gamma(\mathbf{gray}, \ltrue)$, which is not satisfied by $\relint{CD}$.

We aim to obtain a similar result as (\ref{HMext}) above, for a weaker logic introduced in the next section.

%% file: SlcsE.tex
 \section{\slcsE{}: Weak \slcs{} on Polyhedral Models}\label{sec:slcsE}
 
 In this section we introduce \slcsE{,} a logic for polyhedral models that is weaker
 than \slcsG{,} yet is still capable of expressing interesting conditional reachability properties. We present also an interpretation of the logic 
 on finite poset models.

\begin{definition}[Weak \slcs{} on polyhedral models - \slcsE{}]\label{def:SlcsEPolMod}
The abstract language of \slcsE{} is the following:
$$
      \form ::= p \; \sep \; \lneg \form \; \sep \; \form_1 \land \form_2 \; \sep \; \eta(\form_1,\form_2). 
$$
The satisfaction relation of \slcsE{} with respect to a given polyhedral model
$\calP = (|K|, V)$,  \slcsE{} formula $\form$, and point $x \in |K|$ is defined recursively on the structure of $\form$ as 
follows:
$$
\begin{array}{l c l}
\calP, x \models p & \Leftrightarrow & x \in V(p);\\
\calP, x \models \lneg \form & \Leftrightarrow & \calP, x \models  \form \mbox{ does not hold};\\
\calP, x \models \form_1 \land \form_2 & \Leftrightarrow &
\calP, x \models \form_1 \mbox{ and } \calP, x \models \form_2;\\
\calP, x \models \eta(\form_1,\form_2) & \Leftrightarrow &
\mbox{a } \mbox{topological path } \pi: [0,1] \to |K| \mbox{ exists such that } \pi(0)=x,\\&&
\calP, \pi(1) \models \form_2, \mbox{and }
\calP, \pi(r) \models \form_1 \mbox{ for all } r\in \!\![0,\!1).
\end{array}
$$
\vspace{-0.3in}\\\mbox{ }\hfill\closedefi
\end{definition}

As usual, disjunction ($\lor$) is derived as the dual of $\land$.
Note that the only difference between $\eta(\form_1,\form_2)$ and 
$\gamma(\form_1,\form_2)$ is that the former requires that  {\em also the 
first element} of a path witnessing the formula satisfies $\form_1$,
hence the use of the left closed interval $[0,1)$ here.
Although this might seem at first sight only a very minor difference, it has considerable consequences of both theoretical and practical nature, as we will see in what follows.
  
\begin{definition}[\slcsE{} Logical Equivalence]\label{def:SlcsEeq}
Given polyhedral model $\calP = (|K|, V)$ and $x_1, x_2 \in |K|$ we say that $x_1$ and $x_2$ are {\em logically equivalent} with respect to \slcsE{,}  
written $x_1 \slcsEeq^{\calP} x_2$, iff, for all \slcsE{} formulas $\form$, the following holds:
$\calP,x_1 \models \form$ if and only if
$\calP,x_2 \models \form$.\closedefi
\end{definition}

In the sequel, we will refrain from indicating the model $\calP$ explicitly as a superscript of  $ \slcsEeq^{\calP}$ when it is clear from the context.
Below, we show that \slcsE{} can be encoded into \slcsG{} so that 
the former is weaker than the latter.

\begin{definition}\label{def:etga}
We define the following encoding of \slcsE{} into \slcsG{:}
$$
\begin{array}{l c l}
\etga(p) & = & p\\
\etga(\lneg \form) & = & \lneg\etga(\form)
\end{array}\quad\quad\quad\quad
\begin{array}{l c l}
\etga(\form_1 \land \form_2) & = &\etga(\form_1) \land \etga(\form_2)\\
\etga(\eta(\form_1,\form_2)) & = & \etga(\form_1) \land \gamma(\etga(\form_1),\etga(\form_2)) 
\end{array}
$$
\vspace{-0.3in}\\\mbox{ }\hfill\closedefi
\end{definition}

The following lemma is easily proven by structural induction:

\begin{lemma}\label{lem:etgaCorrectG}
Let $\calP=(|K|,V)$ be a polyhedral model, $x \in |K|$ and $\form$ a \slcsE{} formula.
Then $\calP,x \models \form$ iff $\calP,x \models \etga(\form)$.
\end{lemma}

A direct consequence of Lemma~\ref{lem:etgaCorrectG} is that \slcsE{} is weaker than \slcsG{.}

\begin{proposition}\label{prop:SLCSGimplWSLCSG}
Let $\calP=(|K|,V)$ be a polyhedral model.
For all $x_1,x_2 \in |K|$ the following holds:
if $x_1 \slcsGeq{} x_2$ then $x_1 \slcsEeq x_2$.
\end{proposition}

\begin{remark}\label{rem:weakGweakerG}
The converse of Proposition~\ref{prop:SLCSGimplWSLCSG} does {\em not} hold,
as shown by the example polyhedral model $\calP=(|K|,V)$ of Figure~\ref{fig:AlternatingTriangle}.
It is easy to see that, for all  $x\in \relint{ABC}$, we have $A \not\slcsGeq{} x$ and $A \slcsEeq{} x$.
Let $x\in \relint{ABC}$. Clearly, $A \not\slcsGeq{} x$ since $\calP, A \models \gamma(\mathtt{red},\ltrue)$ whereas $\calP, x \not\models \gamma(\mathtt{red},\ltrue)$. 
It can be easily shown, by  induction on the structure of formulas, that $A \slcsEeq x$ for all $x\in \relint{ABC}$ (see Sect.~\ref{apx:prf:rem:weakGweakerG}). 
\closerem
\end{remark}

\begin{figure}
\begin{center}
\subfloat[]{\label{fig:AlternatingTriangle}
\resizebox{0.9in}{!}
{
\begin{tikzpicture}[scale=2]%
    \tikzstyle{point}=[circle,fill=white,inner sep=0pt,minimum width=5pt,minimum height=2pt]	  

            \node (tA)[point, draw=blue!35, fill=blue!35, label={270:$A$}] at (0,0) {};
	    \node (tB)[point, draw=blue!35, fill=blue!35,right of = tA, xshift=2cm, label={270:$B$}]{};
	    \node (tC)[point, draw=blue!35, fill=blue!35,above of = tA, xshift=1.7cm, yshift=1cm, label={90:$C$}]{};
	    \fill [fill=blue!35](tA.center) -- (tB.center) -- (tC.center);
	    \draw [thick, red!60] (tA) -- (tB);	
	    \draw [thick, red!60] (tA) -- (tC);	
	    \draw [thick, red!60] (tB) -- (tC);
 \end{tikzpicture}
 }
 }\quad\quad\quad\quad\quad\quad
 \subfloat[]{\label{fig:PosetAlternatingTriangle}
 \resizebox{1.2in}{!}
 {
 \begin{tikzpicture}[scale=2, every node/.style={transform shape}]
    \tikzstyle{kstate}=[rectangle,draw=black,fill=white]
    \tikzset{->-/.style={decoration={
        markings,
        mark=at position #1 with {\arrow{>}}},postaction={decorate}}}
    \node[kstate,fill=blue!35] (A) at (  0,0) {$\relint{A}$};
    \node[kstate,fill=blue!35, right of= A, xshift=1cm] (C) {$\relint{C}$};
    \node[kstate,fill=blue!35, right of= C, xshift=1cm] (B) {$\relint{B}$};
    \node[kstate,fill=red!60, above of= A] (AC) {$\relint{AC}$};
    \node[kstate,fill=red!60, right of= AC, xshift=1cm] (AB) {$\relint{AB}$};
    \node[kstate,fill=red!60, right of= AB, xshift=1cm] (BC) {$\relint{BC}$};
    \node[kstate,fill=blue!35, above of= AB] (ABC) {$\relint{ABC}$};

    \draw(A) edge[thick](AC);
    \draw(A) edge[thick](AB);
    \draw(C) edge[thick](AC);
    \draw(C) edge[thick](BC);
    \draw(B) edge[thick](AB);
    \draw(B) edge[thick](BC);
    \draw(AC) edge[thick](ABC);
    \draw(AB) edge[thick](ABC);
    \draw(BC) edge[thick](ABC);
\end{tikzpicture}
}
}
\end{center}
\caption{A polyhedral model (\ref{fig:AlternatingTriangle}) and its cell poset model (\ref{fig:PosetAlternatingTriangle})}
\end{figure}
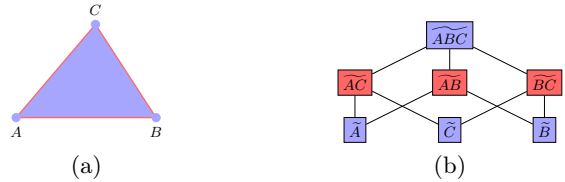

\begin{remark}
The example of Figure~\ref{fig:AlternatingTriangle} is useful also for showing that the classical  topological interpretation of the modal logic operator $\Diamond$ cannot be expressed in \slcsE{.} We recall that
$$
\begin{array}{l c l}
\calP, x \models \Diamond \,\form & \Leftrightarrow & x \in \closure_T(\ZET{x'\in|K|}{\calP,x' \models \form}).
\end{array}
$$
Clearly, in the model of the figure, we have $\calP,A \models \Diamond \mathbf{red}$ while $\calP, x \models \Diamond \mathbf{red}$ for no $x\in \relint{ABC}$. On the other hand, $A \slcsEeq x$ holds for all $x \in \relint{ABC}$, as we have just seen in Remark~\ref{rem:weakGweakerG}. So, if
$\Diamond$ were expressible in \slcsE{,} then $A$ and $x$ should have agreed on $\Diamond \mathbf{red}$ for each $x\in \relint{ABC}$.
Note that $\Diamond$ {\em can} be expressed in \slcsG{} as $\gamma (\form, \ltrue) $, see~\cite{Be+22}.
\end{remark}

Below we re-interpret \slcsE{} on finite poset models instead of polyhedral models. The only difference from Def.~\ref{def:SlcsEPolMod} is, of course, the fact that 
$\eta$-formulas are defined for \plm-paths instead of topological ones.

\begin{definition}[\slcsE{} on finite poset models]\label{def:WSlcsFinPos}
The satisfaction relation of \slcsE{} with respect to a given  finite poset model
$\calF = (W,{\preceq},\peval)$, 
\slcsE{} formula $\form$, and element $w \in W$ is defined recursively on the structure of $\form$:\\

\noindent
$$
\begin{array}{l c l}
\calF, w \models p & \Leftrightarrow & w \in \peval(p);\\
\calF, w \models \lneg \form & \Leftrightarrow & \calF, w \not\models  \form;\\
\calF, w \models \form_1 \land \form_2 & \Leftrightarrow &
\calF, w \models \form_1 \mbox{ and } \calF, w \models \form_2;\\
\calF, w \models \eta(\form_1,\form_2) & \Leftrightarrow &
\mbox{a \plm-path } \pi: [0;\ell] \to W \mbox{ exists such that }
\pi(0)=w, \\&&
\calF, \pi(\ell) \models \form_2 \mbox{ and } 
\calF, \pi(i) \models \form_1 \mbox{ for all } i\in [0;\ell).
\end{array}
$$
\vspace{-0.3in}\\\mbox{ }\hfill\closedefi		
\end{definition}

\begin{definition}[Logical Equivalence]\label{def:WFPSLCSeq}
Given  a finite poset model  $\calF = (W,{\preceq},\peval)$ and elements $w_1, w_2 \in W$ 
we say that $w_1$ and $w_2$ are {\em logically equivalent} with respect to \slcsE, 
written $w_1 \slcsEeq^{\calF} w_2$, iff, for all \slcsE{} formulas $\form$, the following holds:
$\calF,w_1 \models \form$ if and only if
$\calF,w_2 \models \form$.\closedefi
\end{definition}

Again, in the sequel, we will refrain from indicating the model $\calF$ explicitly in $ \slcsEeq^{\calF}$ when it is clear from the context. 
It is useful to define a ``characteristic''  \slcsE{} formula $\chi(w)$ that is satisfied by all and only those $w'$ with $w' \slcsEeq w$. 

\begin{definition}\label{def:chiE}
Given a finite poset model $(W,\preceq,\peval)$,
$w_1,w_2 \in W$, define \slcsE{} formula~$\delta_{w_1,w_2}$
  as follows: if $w_1 \slcsEeq w_2$, then set $\delta_{w_1,w_2}= \ltrue$, otherwise pick some 
  \slcsE{} formula~$\psi$ such that $\calF,w_1 \models \psi$ and $\calF,w_2 \models \lneg\psi$, and  set $\delta_{w_1,w_2}= \psi$.
For $w \in W$ define $\chi(w) = \bigwedge_{w' \in W} \: \delta_{w,w'}$.
\closedefi
\end{definition}

\begin{proposition}\label{prop:chiE} 
Given a finite poset model $(W,\preceq,\peval)$, for $w_1,w_2 \in W$, it holds that 
 $\calF,w_2 \models \chi(w_1) \mbox{ if and only if } w_1 \slcsEeq w_2$.
  \end{proposition}

The following lemma is the poset model counterpart of Lemma~\ref{lem:etgaCorrectG}:

\begin{lemma}\label{lem:etgaCorrectE}
Let $\calF=(W,\preceq,\peval)$ be a 
finite poset model, 
$w \in W$ and $\form$ a \slcsE{} formula.
Then $\calF,w \models \form$ iff $\calF,w \models \etga(\form)$.
\end{lemma}

Thus we get, as for the interpretation on polyhedral models, that 
\slcsE{} on finite poset models is weaker than \slcsG{:}

\begin{proposition}\label{prop:SLCSPMimplWSLCSPM}
Let $\calF=(W,\preceq, \peval)$ be a finite poset model. 
For all $w_1,w_2 \in W$ the following holds:
if $w_1 \slcsGeq{} w_2$ then $w_1 \slcsEeq w_2$.
\qed
\end{proposition}

\begin{remark}\label{rem:weakPMeakerPM}
As expected, the converse of Proposition~\ref{prop:SLCSPMimplWSLCSPM} does not hold,
as shown by the poset model $\calF$ of  Figure~\ref{fig:PosetAlternatingTriangle}.
Clearly, $\relint{A} \not\slcsGeq{} \relint{ABC}$. In fact $\calF,\relint{A} \models \gamma(\mathtt{red},\ltrue)$
whereas $\calF,\relint{ABC} \not\models \gamma(\mathtt{red},\ltrue)$. On the other hand, 
$\relint{A} \slcsEeq \relint{ABC}$ can be easily proven by induction on the structure of formulas
(see Sect.~\ref{apx:prf:rem:weakPMeakerPM}).
\closerem
\end{remark}

\begin{remark}\label{rem:noDiamondInEta}
The example of Figure~\ref{fig:PosetAlternatingTriangle}  is useful also for showing that the classical  modal logic operator $\Diamond$ cannot be expressed in \slcsE{.} We recall that 
$$
\begin{array}{l c l}
\calF, w\models \Diamond \form & \Leftrightarrow &
w' \in W \mbox{ exists such that } w \preceq w' \mbox{ and } \calF, w' \models \form.
\end{array}
$$
Clearly, in the model of the figure, we have $\calF, \relint{A} \models \Diamond \mathtt{red}$ while
$\calF, \relint{ABC} \not\models \Diamond \mathtt{red}$. On the other hand $\relint{A} \slcsEeq \relint{ABC}$ holds, as we have just seen in Remark~\ref{rem:weakPMeakerPM}. So, if 
$\Diamond$ were expressible in \slcsE{,} then $\relint{A}$ and $\relint{ABC}$ should have agreed on $\Diamond \mathtt{red}$. 
Note that $\Diamond$ can be expressed in \slcsG, as the following equality can be 
proven~\cite{Be+22}: $\Diamond \form \equiv \gamma(\form,\ltrue)$.
\closerem
\end{remark}

The following result is useful for setting a bridge between the continuous and the discrete interpretation of \slcsE{.}

\begin{lemma}\label{lem:VB}
Given a polyhedral model $\calP=(|K|,V)$,
for all $x\in |K|$ and formulas $\form$ of \slcsE{}
the following holds: $\calP,x \models \form$ iff $\map(\calP),\map(x) \models \etga(\form)$.
\end{lemma}

As a direct consequence of Lemma~\ref{lem:etgaCorrectE} and Lemma~\ref{lem:VB} we get the bridge between the continuous and the discrete interpretation of \slcsE{.}

\begin{theorem}\label{theo:calMPresForm}
Given a polyhedral model $\calP=(|K|,V)$,
for all $x\in |K|$ and formulas $\form$ of \slcsE{}
it holds that: $\calP,x \models \form$ iff $\map(\calP),\map(x) \models \form$.
\end{theorem}

This theorem allows one to go back and forth between the polyhedral model and the corresponding poset model without loosing anything expressible in \slcsE{}.

%% file: WeakBis.tex
 \section{Weak Simplicial Bisimilarity}\label{sec:WeakBis}

In this section, we introduce weak versions of simplicial bisimilarity and \plm-bisimilarity and we show that they coincide with logical equivalence induced by \slcsE{} in polyhedral and poset models, respectively.

\begin{definition}[Weak Simplicial Bisimulation]\label{def:WSimBis}
Given a polyhedral model $\calP=(|K|,V)$,  a symmetric relation $B \subseteq |K| \times |K|$
is a {\em weak simplicial bisimulation} if, for all $x_1,x_2 \in |K|$,
whenever $B(x_1,x_2)$, it holds that:
\begin{enumerate}
\item $V^{-1}(\SET{x_1}) = V^{-1}(\SET{x_2})$;
\item 
for each topological path $\pi_1$ from $x_1$, 
there is topological path $\pi_2$ from $x_2$ 
such that $B(\pi_1(1),\pi_2(1))$ and
for all $r_2 \in [0,1)$ there is $r_1 \in  [0,1)$ such that $B(\pi_1(r_1),\pi_2(r_2))$.
\end{enumerate}
Two points $x_1,x_2 \in |K|$ are weakly simplicial bisimilar, written $x_1 \wsibis^{\calP} x_2$, if there is
a weak simplicial bisimulation  $B$ such that $B(x_1,x_2)$.
\end{definition}

For example, the open segments $AB$, $BC$, and $AC$ in Figure~\ref{fig:AlternatingTriangle} are mutually weakly simplicial bisimilar and every point in set 
$\relint{ABC} \cup \relint{A} \cup \relint{B} \cup \relint{C}$ 
is weakly simplicial bisimilar
to every other point in the same set.

\begin{definition}[Weak \plm-bisimulation]\label{def:WPLMBis}
Given a finite poset model $\calF=(W,\preceq~\!\!,\peval)$, a symmetric binary relation $B \subseteq W \times W$
is a weak  \plm-bisimulation if, for all $w_1,w_2 \in W$,
whenever $B(w_1,w_2)$, it holds that:
\begin{enumerate}
\item $\invpeval(\SET{w_1}) = \invpeval(\SET{w_2})$;
\item 
for each $u_1,d_1 \in W$ such that $w_1 \; \preceqsucc \; u_1 \succeq \, d_1$
there is a \plm-path $\pi_2:[0;\ell_2] \to W$ from $w_2$ such that
$B(d_1,\pi_2(\ell_2))$ and, for all $j\in[0;\ell_2)$,  the following holds: 
$B(w_1,\pi_2(j))$ or $B(u_1,\pi_2(j))$. 
\end{enumerate}
We say that $w_1$ is weakly \plm-bisimilar to $w_2$, written $w_1 \wplmbis^{\calF} w_2$ if
there is a weak  \plm-bisimulation $B$ such that $B(w_1,w_2)$.\closedefi
\end{definition}

For example, all red cells in the Hasse diagram of Figure~\ref{fig:PosetAlternatingTriangle} are weakly \plm-bisimilar and all blue cells are weakly \plm-bisimilar.

The following lemma shows that, in a polyhedral model $\calP$, weak simplicial bisimilarity $\wsibis^{\calP}$ (Def.~\ref{def:WSimBis}) is stronger than $\slcsEeq$ --
logical equivalence w.r.t. \slcsE{:} 

\begin{lemma}\label{lem:WSBisIMPLEtaLog}
Given a polyhedral model  $\calP=(|K|,V)$, for all $x_1,x_2 \in |K|$, the following holds:
if $x_1 \wsibis^{\calP} x_2$ then  $x_1 \slcsEeq x_2$.
\end{lemma}

\begin{proof}
By induction on the structure of the formulas. We consider only the case $\eta(\form_1, \form_2)$.
Suppose $x_1 \wsibis x_2$ and $\calP,x_1 \models \eta(\form_1, \form_2)$.
Then there is a topological path $\pi_1$ from $x_1$ such that 
$\calP,\pi_1(1)\models \form_2$ and $\calP,\pi_1(r_1)\models \form_1$ for all $r_1\in [0,1)$.
Since $x_1 \wsibis x_2$, then there is  a topological path $\pi_2$ from $x_2$ such that
$\pi_1(1) \wsibis \pi_2(1)$ and for each  $r_2 \in [0,1)$ there is $r'_1 \in [0,1)$
such that $\pi_1(r'_1) \wsibis \pi_2(r_2)$.
By the Induction Hypothesis, we get $\calP, \pi_2(1)\models \form_2$ 
and for each  $r_2 \in [0,1)$ 
$\calP, \pi_2(r_2) \models  \form_1$. Thus $\calP,x_2 \models \eta(\form_1, \form_2)$.\qed
\end{proof}

Furthermore, logical equivalence induced by \slcsE{} is stronger than weak simplicial-bisimilarity, 
as implied by Lemma~\ref{lem:EtaLogIsWSB} below, which uses the following two auxiliary lemmas:

\begin{lemma}\label{lem:dPExists}
Given a finite poset model $\calF=(W,\preceq, \peval)$ and  
weak \plm-bisimulation $B \subseteq W \times W$, for all $w_1,w_2$ such that $B(w_1,w_2)$, the following holds:
for each  \dwn-path $\pi_1:[0;k_1] \to W$ from $w_1$ 
there is a \dwn-path $\pi_2:[0;k_2] \to W$ from $w_2$ 
such that $B(\pi_1(k_1),\pi_2(k_2))$ and 
for each $j\in [0;k_2)$ there is $i\in [0;k_1)$ such that $B(\pi_1(i),\pi_2(j))$.
\end{lemma}

\begin{lemma}\label{lem:dPtoTP}
Given a polyhedral model $\calP=(|K|,V)$, and associated cell poset model $\map(\calP)=(W,\preceq,\peval)$, for any 
\dwn-path  $\pi:[0;\ell] \to W$, 
there is a topological path $\pi':[0,1] \to |K|$ such that: (i) $\map(\pi'(0))=\pi(0)$, (ii) $\map(\pi'(1))=\pi(\ell)$, and
(iii) for all $r \in (0,1)$ there is $i<\ell$ such that $\map(\pi'(r))=\pi(i)$.
\end{lemma}

\begin{lemma}\label{lem:EtaLogIsWSB}
In a given polyhedral model  $(|K|,V)$, $\slcsEeq$ is a weak simplicial bisimulation.
\end{lemma}

\begin{proof}
Let $x_1,x_2 \in |K|$ such that $x_1 \slcsEeq x_2$. 
The first condition of Definition~\ref{def:WSimBis} is clearly satisfied since $x_1 \slcsEeq x_2$.
Suppose $\pi_1$ is a topological path from $x_1$.
$\map(\pi_1([0,1]))$ is a connected subposet of $\relint{K}$. 
Thus, due to continuity of $\map\circ\pi_1$, a \dwn-path $\hat\pi_1:[0;k_1] \to \relint{K}$  
from $\map(\pi_1(0))$ to $\map(\pi_1(1))$ exists such that for all $i\in [0;k_1)$ there is
$r_1 \in [0,1)$ with $\hat\pi_1(i)=\map(\pi_1(r_1))$. 
We also know that $\map(x_1) \slcsEeq \map(x_2)$, as a consequence of  Theorem~\ref{theo:calMPresForm},
since $x_1 \slcsEeq x_2$. In addition, due to Lemma~\ref{lem:EtaLogIsWFpBis} below, we also
know that $\map(x_1) \wplmbis \map(x_2)$. By Lemma~\ref{lem:dPExists},  we get that 
there is a \dwn-path $\hat\pi_2:[0;k_2] \to \relint{K}$  
such that 
$\hat\pi_1(k_1) \slcsEeq \hat\pi_2(k_2)$ and for each $j\in [0;k_2)$ there is $i\in [0;k_1)$
such that $\hat\pi_1(i) \slcsEeq \hat\pi_2(j)$. 
By Lemma~\ref{lem:dPtoTP}, it follows that there is topological path $\pi_2$ from $x_2$
satisfying the three conditions of the lemma and, again by Theorem~\ref{theo:calMPresForm},
we have that $\pi_2(1) \slcsEeq \pi_1(1)$. In addition, for any $r_2 \in [0,1)$, since 
$\map(\pi_2(r_2))= \hat\pi_2(j)$ for $j\in [0;k_2)$ (condition (ii) of Lemma~\ref{lem:dPtoTP})
there is $i\in[0;k_1)$ such that $\hat\pi_1(i) \slcsEeq \hat\pi_2(j)$.
Finally, by construction, there is $r_1\in [0,1)$ such that $\map(\pi_1(r_1)) = \hat\pi_1(i)$.
By Theorem~\ref{theo:calMPresForm}, we finally get
$\pi_1(r_1) \slcsEeq \pi_2(r_2)$.\qed
\end{proof}

On the basis of Lemma~\ref{lem:WSBisIMPLEtaLog} and Lemma~\ref{lem:EtaLogIsWSB}, we have that the largest weak simplicial bisimulation exists, it is a weak simplicial bisimilarity, it is an equivalence relation, and it coincides with logical equivalence in the polyhedral model induced by \slcsE{,} thus establishing the HMP for $\wsibis^{\calP}$ w.r.t.  \slcsE{:}

\begin{theorem}\label{thm:HMPpoly}
Given a polyhedral model $\calP= (|K|,V), x_1,x_2 \in |K|$, the following holds:
$x_1 \slcsEeq^{\calP} x_2$
iff
$x_1 \wsibis^{\calP} w_2$. \qed
\end{theorem}

Similar results can be obtained for poset models. The following lemma shows that, in every finite poset model $\calF$, weak \plm-bisimilarity (Def.~\ref{def:WPLMBis}) is stronger than
logical equivalence with respect to \slcsE{,} i.e. $\wplmbis^{\calF} \: \subseteq \: \slcsEeq^{\calF}$:

\begin{lemma}\label{lem:WFpBisIMPLEtaLog}
Given a finite poset model $\calF=(W,\preceq, \peval)$, for all $w_1,w_2 \in W$,
if $w_1 \wplmbis^{\calF} w_2$ then  $w_1 \slcsEeq^{\calF} w_2$.
\end{lemma}

\begin{proof}
By induction on formulas. We consider only the case $\eta(\form_1,\form_2)$.
Suppose $w_1 \wplmbis w_2$ and $\calF, w_1 \models \eta(\form_1,\form_2)$.
Then, there is (a \plm-path and so) a \dwn-path $\pi_1$ from $w_1$ of some length $k_1$ such that 
$\calF,\pi_1(k_1) \models \form_2$ and for all $i\in [0;k_1)$ $\calF,\pi_1(i) \models \form_1$ holds.
By Lemma~\ref{lem:dPExists}, we know that a \dwn-path $\pi_2$ from $w_2$ exists 
of some length $k_2$ such that $\pi_1(k_1) \wplmbis \pi_2(k_2)$ and
for all $j \in [0;k_2)$ there is $i\in [0;k_1)$ such that $\pi_1(i) \wplmbis \pi_2(j)$.
By the Induction Hypothesis, we then get that $\calF,\pi_2(k_2) \models \form_2$ and
for all $j \in [0;k_2)$ we have $\calF,\pi_2(j) \models \form_1$. This implies that 
$\calF, w_2 \models \eta(\form_1,\form_2)$.
\end{proof}

Furthermore, logical equivalence induced by \slcsE{} is stronger than weak \plm-bisimilarity, i.e.
$\slcsEeq^{\calF}\: \subseteq\: \wplmbis^{\calF}$, as implied by the following:

\begin{lemma}\label{lem:EtaLogIsWFpBis}
In a finite poset model $\calF=(W,\preceq, \peval)$, $\slcsEeq^{\calF}$ is a weak \plm-bisimulation.
\end{lemma}

\begin{proof}
If $w_1 \slcsEeq w_2$, then the first requirement of Definition~\ref{def:WPLMBis} is trivially satisfied.
We prove that $\slcsEeq$ satisfies the second requirement of Definition~\ref{def:WPLMBis}.
Suppose $w_1 \slcsEeq w_2$ and let $u_1, d_1$  as in the above mentioned requirement.
This implies that $\calF, w_1 \models \eta(\chi(w_1) \lor \chi(u_1), \chi(d_1))$, where,
we recall, $\chi(w)$ is the `characteristic formula' for $w$ as in Definition~\ref{def:chiE}.
Since $w_1 \slcsEeq w_2$, we have that also
 $\calF, w_2 \models \eta(\chi(w_1) \lor \chi(u_1), \chi(d_1)$ holds. This in turn means that 
 a \dwn-path $\pi_2$ of some length $k_2$ from $w_2$ exists such that 
 $\calF, \pi_2(k_2) \models \chi(d_1)$ and
 for all $j\in [0;k_2)$ we have $\calF, \pi_2(j) \models \chi(w_1) \lor \chi(u_1)$,
 i.e. $\calF, \pi_2(j) \models \chi(w_1)$ 
 or $\calF, \pi_2(j) \models \chi(u_1)$.
 Consequently, by Proposition~\ref{prop:chiE}, we have:
 $\pi_2(k_2) \slcsEeq d_1$ and,
 for all $j\in [0;k_2)$, $\pi_2(j) \slcsEeq w_1$ or $\pi_2(j) \slcsEeq u_1$, 
so that the second condition of the definition is fulfilled.
\end{proof}

On the basis of Lemma~\ref{lem:WFpBisIMPLEtaLog} and Lemma~\ref{lem:EtaLogIsWFpBis}, we have that the largest weak \plm-bisimulation exists, it is a weak \plm-bisimilarity, it is an equivalence relation, and it coincides with logical equivalence in the finite poset induced by \slcsE{:}

\begin{theorem}\label{theo:PMbisEqSLCSEeq}
For every  finite poset model $(W,{\preceq},\peval), w_1,w_2 \in W$, the following holds:
$w_1 \slcsEeq^{\calF} w_2$ 
iff
$w_1 \wplmbis^{\calF} w_2$.
\end{theorem}
By this we have established the HMP for $\wplmbis$ w.r.t. \slcsE{.}

Finally, recalling that, by Theorem~\ref{theo:calMPresForm}, given polyhedral model $\calP =(|K|,V)$ for all $x\in |K|$ and \slcsE{} formula $\form$, we have that 
$\calP,x \models \form$ if and only if $\map(\calP),\map(x) \models \form$, we get the following final result:

\begin{corollary}\label{cor:coincidenceE}
For all  polyhedral models $\calP=(|K|,V)$, $x_1,x_2 \in |K|$:
$$
x_1 \wsibis^{\calP} x_2 \mbox{ iff }
x_1  \slcsEeq^{\calP}  x_2 \mbox{ iff }
\map(x_1)  \slcsEeq^{\map(\calP)} \map(x_2) \mbox{ iff }
\map(x_1)  \wplmbis^{\map(\calP)}  \map(x_2).
$$
\end{corollary}

\begin{figure}[t!]
\begin{center}
\begin{tikzpicture}[scale=0.6, every node/.style={transform shape}]
    \tikzstyle{gstate}=[circle,draw=black,fill=white]
    \tikzset{->-/.style={decoration={
        markings,
        mark=at position #1 with {\arrow{>}}},postaction={decorate}}}
    \node[gstate,fill=gray!50] (C9) at ( 0,0) {$\mathbb{C}'_1$};
    \node[gstate,fill=red!50,right of=C9,xshift=3cm] (C013){$\mathbb{C}'_2$};
    \node[gstate,fill=green!50,above of=C013,xshift=2.2cm,yshift=1.5cm] (C4){$\mathbb{C}'_4$};
    \node[gstate,fill=gray!50,right of=C013,xshift=3cm] (C25678){$\mathbb{C}'_3$};
    \draw(C9) edge[->](C013);
    \draw(C9) edge[->, loop below](C9);
    \draw(C4) edge[->, loop above](C4);
    \draw(C013) edge[->](C4);
    \draw(C25678) edge[->](C4);
    \draw(C25678) edge[->, bend right](C013);
    \draw(C013) edge[->, loop below](C013);
    \draw(C013) edge[->, bend right](C25678);
    \draw(C25678) edge[->, loop below](C25678);
\end{tikzpicture}
\end{center}
\caption{The minimal model, modulo weak \plm-bisimilarity, of the model of Fig.~\ref{fig:PolyhedronNoPathCompressed}.
}\label{fig:exa:MinRunExaE}
\end{figure}
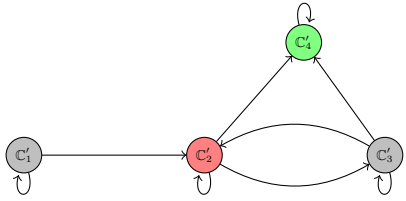
\noindent Saying that \slcsE-equivalence in a polyhedral model is the same as weak simplicial bisimilarity, which maps by $\map$ to the weak \plm-bisimilarity in the corresponding poset model, where the latter coincides with the \slcsE-equivalence. 
\begin{example}\label{exa:MinRunExaE}
Fig.~\ref{fig:exa:MinRunExaE} shows the minimal model 
$\min(\map(\calP))$, modulo $\wplmbis$, of $\map(\calP)$  (see Fig.~\ref{subfig:PolyhedronNoPathPosetCompressed}).
We have the following equivalence classes:
$\mathbb{C}'_1=\SET{\relint{A}}$, 
$\mathbb{C}'_2=\SET{\relint{B}, \relint{C}, \relint{AB}, \relint{AC}, \relint{BC}, \relint{BD}, \relint{CD},
\relint{ABC}, \relint{BCD}}$, 
$\mathbb{C}'_3=\SET{\relint{D},\relint{E},\relint{F},\relint{CE},\relint{DE},\relint{DF},\relint{EF},\\
\relint{DEF}}$ and
$\mathbb{C}'_4=\SET{\relint{CDE}}$.
Note that the minimal model is not a poset model, but it is a reflexive Kripke model.\footnote{It is worth noting that for model-checking purposes we can safely interpret \slcsE{} over (reflexive) Kripke models. The satisfaction relation is defined as in Def.~\ref{def:WSlcsFinPos} where $\calF$ is a Kripke model instead of a poset model (recall that \plm-paths are defined on Kripke frames).}
\closeex
\end{example}

%% file: example.tex
\section{A Larger Example}
\label{sec:example}

As a proof-of-concept and feasibility we show a larger example of a 3D polyhedral structure composed of one white "room" and 26 green "rooms" connected by grey "corridors" as shown in Fig.~\ref{subfig:cube}.  In turn, each room is composed of 33 vertices, 122 edges, 150 triangles and 60 tetrahedra, i.e. it is composed of a total of 365 cells. Each corridor is composed of 8 edges, 12 triangles and 5 tetrahedra, i.e. it consists of 25 cells. The corridors are connected to rooms via the four points of the side of a room. In total, the structure consists of 11,205 cells. We have chosen a large, but symmetric structure on purpose. This makes it easy to interpret the various equivalence classes present in the minimal Kripke model of this structure shown in Fig.~\ref{subfig:minKripke}. 
Observe that, for this example, the minimal  model is also a poset model and, in particular,  a cell poset model representing  a polyhedron, as shown in Fig.~\ref{subfig:minPoly}. The latter can be seen as a minimised version of the original polyhedral structure. Note also the considerable reduction that was obtained: from 11,205 cells to just 7 in the minimal model.

In Fig.~\ref{subfig:minKripke} we have indicated the various equivalence classes with a letter. Those indicated with a "C" correspond to classes of (cells of) corridors, those with an "R" correspond to classes of (cells of) rooms.
For reasons of space and clarity, in the following we will not list all the individual cells that are part of a certain class, but instead we will indicate those cells by speaking about certain rooms and corridors, intending the cells that they are composed of.

There is one white class containing all white cells of the white room. Furthermore, there are three green classes corresponding to three types of green rooms, and three grey classes corresponding to three kinds of corridors. The green class R2 is composed of the (cells in) the six green rooms situated in the middle of each side of the cube structure. Those in R3 are the cells in the twelve green rooms situated in the middle of each `edge' of the cube structure. Those in R4 are the cells in the eight green rooms situated at the corners of the cube structure.
It is not difficult to find \slcs$_\eta$ formulas that distinguish, for instance, the various green classes. 
For example, the cells in R2 satisfy $\form_1=\eta(\mathtt{green} \lor \eta(\mathtt{grey},\mathtt{white}), \mathsf{white})$, whereas no cell in R3 or R4 satisfies~$\form_1$. 
To distinguish class R3 from R4 one can observe that cells in R3 satisfy 
$\form_2= \eta(\mathtt{green} \lor \eta(\mathtt{grey},\form_1), \form_1)$ whereas those in R4 do not satisfy $\form_2$.

In this symmetric case of this synthesised example, it was rather straightforward to find the various equivalence classes. In the general case it is much harder and one would need a suitable minimisation algorithm for \slcsE{.} We are currently working on an effective minimisation procedure based on encoding the cell poset model into a suitable \lts{,} exploiting behavioural equivalences for \lts{s} ---  strong bisimilarity and branching bisimulation equivalence --- following an approach similar to that followed in~\cite{Ci+23a} for finite closure spaces.
The first results are promising. In fact, the large structure shown in this example can be handled that way and gives results as presented. Details and proofs of correctness of this approach and its potential efficiency gain will be the topic of future work, also for reasons of space limitations.

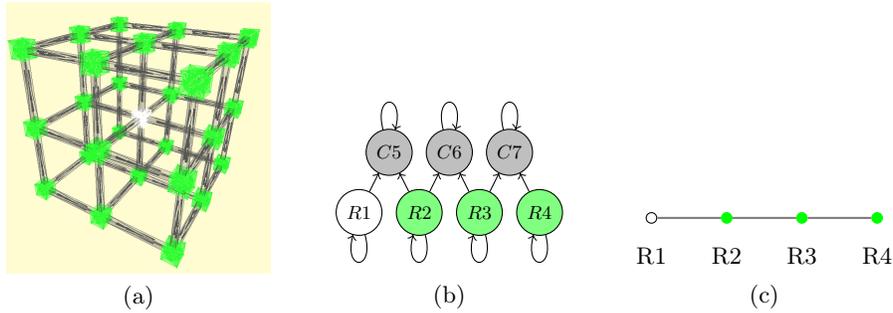
\begin{figure}[t!]
\centering
	\subfloat[]{\label{subfig:cube}
		\includegraphics[valign=b,height=11em]{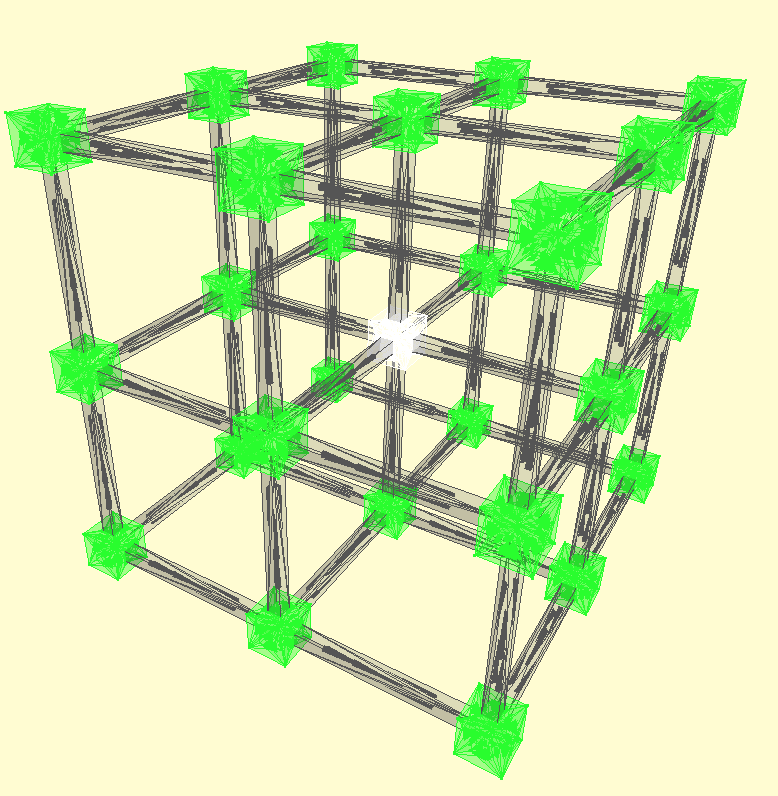}
	}\quad\quad
	\subfloat[]{\label{subfig:minKripke}
    \begin{tikzpicture}[scale=0.8, every node/.style={transform shape}]
    \tikzstyle{kstate}=[circle,draw=black,fill=white]
    \tikzset{->-/.style={decoration={
		markings,
		mark=at position #1 with {\arrow{>}}},postaction={decorate}}}
   

    
    \node[kstate,fill=white!50,label={270:$$}  ] (P0) at (  0,0) {$R1$};
    \node[kstate,fill=green!50,label={270:$$}   ] (P1) at (  1,0) {$R2$};
    \node[kstate,fill=green!50,label={270:$$}   ] (P2) at (  2,0) {$R3$};
    \node[kstate,fill=green!50,label={270:$$}   ] (P3) at (  3,0) {$R4$};

    \node[kstate,fill=gray!50,label={90:$$} ] (E0) at (0.5,1) {$C5$};
    \node[kstate,fill=gray!50,label={90:$$} ] (E1) at (1.5,1) {$C6$};
    \node[kstate,fill=gray!50,label={90:$$} ] (E2) at (2.5,1) {$C7$};

    \draw (P0)[->] to (E0);
    \draw (P1)[->] to (E0);
    \draw (P1)[->] to (E1);

    \draw (P2)[->] to (E1);
    \draw (P2)[->] to (E2);
    \draw (P3)[->] to (E2);
    
    \draw (P0)[->,loop below] to (P0);
    \draw (P1)[->,loop below] to (P1);
    \draw (P2)[->,loop below] to (P2);
    \draw (P3)[->,loop below] to (P3);
    \draw (E0)[->,loop above] to (E0);
    \draw (E1)[->,loop above] to (E1);
    \draw (E2)[->,loop above] to (E2);

\end{tikzpicture}  
	}\quad\quad
	\subfloat[]{\label{subfig:minPoly}
     \begin{tikzpicture}
       \tikzset{node distance=1cm}
       \tikzstyle{point}=[circle,inner sep=0pt,minimum width=4pt,minimum height=4pt]
           \node (A)[point,draw=black,fill=white] at (0,0){};
           \node (Aname) [above of = A,yshift = -1.5cm] {R1};
           \node (B)[point,draw=green,fill=green, right of= A] {};
          \node (Bname) [above of = B,yshift = -1.5cm] {R2};
          \node (C)[point,draw=green,fill=green, right of= B] {};
          \node (Cname) [above of = C,yshift = -1.5cm] {R3};
           \node (D)[point,draw=green,fill=green, right of= C]{};
          \node (Dname) [above of = D,yshift = -1.5cm] {R4};
       
          \draw [thick,gray] (A) edge (B);
          \draw [thick,gray] (B) edge (C);
          \draw [thick,gray] (C) edge (D);
     \end{tikzpicture}
	}
	
\caption{ (\ref{subfig:cube}) A simplicial complex of a 3D structure composed of rooms and corridors. (\ref{subfig:minKripke}) Its minimal Kripke structure. (\ref{subfig:minPoly}) Its minimal polyhedron.}\label{fig:cube}
\end{figure}

%% file: conclusions.tex
\section{Conclusions}
\label{sec:ConclusionsFW}

In~\cite{Be+22} simplicial bisimilarity was proposed for polyhedral models --- i.e. models of continuous space --- while  \plm-bisimilarity, the corresponding equivalence for cell-poset models --- discrete representations of polyhedral models --- was introduced in~\cite{Ci+23c}.
In order to support large model reductions, in this paper the novel notions of weak  simplicial bisimilarity and  weak  \plm-bisimilarity have been proposed, and the correspondence between the two has been studied. 
We have proposed \slcsE{,} a weaker version of the Spatial Logic for Closure Spaces on polyhedral models, and we have shown that simplicial bisimilarity enjoys the Hennessy-Milner property (Thm.~\ref{thm:HMPpoly}). We have also proven that the property holds for \plm-bisimilarity on poset models and the interpretation of \slcsE{} on such models (Thm.~\ref{theo:PMbisEqSLCSEeq}).
\slcsE{} can be used in the geometric spatial model checker  \polylogica{} for checking spatial reachability properties of polyhedral models. 
Model checking results can be visualised by projecting them onto the original polyhedral structure in a colour.
The results presented in this paper also have a practical value for the domain of visual computing where polyhedral models can be found in the  form of surface meshes or tetrahedral volume meshes that are often composed of a huge number of cells.

In future work, in line with our earlier work, we aim to develop an automatic, provably correct minimisation procedure so that model checking could potentially be performed on a much smaller model. We also intend to develop a procedure to translate results back to the original polyhedral model for their appropriate visualisation. Finally, the complexity and efficiency of such methods will be investigated.\vspace{-0.08in}

%% file: Appendix.tex
\section{Background and Notation in Detail}\label{apx:BackgroundInDetail}

For sets $X$ and $Y$, a function $f:X \to Y$, and subsets $A \subseteq X$ and $B \subseteq Y$, 
we define $f(A)$ and $f^{-1}(B)$ as $\ZET{f(a)}{a \in A}$  and $\ZET{a}{f(a) \in B}$, respectively. 
The  {\em restriction} of  $f$ on $A$ is denoted by $f|A$.
The powerset of $X$ is denoted by $\pws{X}$.
For relation $R\subseteq X\times X$ we let $\cnv{R}$ denote its converse and $\dircnv{R}$ denote $R \, \cup  \cnv{R}$.
In the sequel, we assume that a set  $\ap$ of {\em proposition letters} is fixed.
The set of natural numbers and that of real numbers are denoted by $\nats$ and $\reals$ respectively. 
We use the standard interval notation: for $x,y \in \reals$ we let $[x,y]$ be the set
$\ZET{r\in \reals}{x\leq r \leq y}$, $[x,y) = \ZET{r\in \reals}{x\leq r < y}$ and so on,
where $[x,y]$ is equipped with the Euclidean topology inherited from $\reals$.
We use a similar notation for intervals over $\nats$:  
for $n,m \in \nats$ 
$[m;n]$ denotes the set $\ZET{i\in\nats}{m\leq i \leq n}$,
$[m;n)$ denotes the set $\ZET{i\in\nats}{m\leq i < n}$, and similarly for
$(m;n]$ and $(m;n)$.

A {\em topological space} is a pair $(X,\tau)$ where $X$ is a set (of {\em points}) and $\tau$  is a collection of subsets
of $X$ satisfying the following axioms: 
(i) $\emptyset,X \in \tau$,
(ii) for any index set $I$, $\bigcup_{i\in I} A_i \in \tau$ if each $A_i \in \tau$, and
(iii) for any finite index set $I$, $\bigcap_{i\in I} A_i \in \tau$ if each $A_i \in \tau$.
We let $\closure_T$ denote the topological closure operator.

A {\em Kripke frame} is a pair $(W, R)$ where $W$ is a set and $R \subseteq W \times W$, the {\em accessibility} relation on $W$. 
A {\em Kripke model} is a tuple $(W,R,\peval)$ where $(W,R)$ is a Kripke frame and $\peval: \ap \to \pws{W}$ is the valuation function, assigning to each $p\in \ap$ the set $\peval(p)$ of elements of $W$ where $p$ holds.

In the context of the present paper, it is convenient to view a 
{\em partially ordered set} --- poset, in the sequel --- $(W,\preceq)$ as a Kripke frame where the 
relation $\preceq \, \subseteq W \times W$ is a partial order, i.e. it is reflexive, transitive and anti-symmetric.
Similarly we define a {\em poset model} as a Kripke model  where the accessibility relation is a partial order. For partial orders $(W,\preceq)$, we use the standard notation, i.e.: $\cnv{\preceq}$ will be denoted by $\succeq$, $w_1 \prec w_2$ denotes $w_1 \preceq w_2$ and $w_1 \not= w_2$, and similarly for $\succ$.

\subsection{Paths}
Paths play a crucial role in the present paper. In the sequel, we provide definitions for the different kinds of paths we will use later on in the paper and we prove some useful properties of theirs.

\begin{definition}[Topological Path]\label{def:TopologicalPath}
Given a topological space $(X,\tau)$ and $x\in X$, a {\em topological path} from $x$  is a total, continuous function $\pi : [0,1] \to X$ such that $\pi(0)=x$.
We call $x$ the {\em starting point} of $\pi$. The {\em ending point} of $\pi$ is $\pi(1)$, while
for any $r\in (0,1)$, $\pi(r)$ is an {\em intermediate point} of $\pi$.
\closedefi
\end{definition}

\begin{definition}[Paths Over Kripke Frames]\label{def:PathsOverKripkeFrames}
Given a Kripke frame $(W,R)$ and $w \in W$:
\begin{itemize}
\item 
An {\em undirected path} from $w$, of length $\ell \in \nats$,  is  a total function $\pi : [0;\ell] \to W$ such that $\pi(0)=w$ and, for all $i\in [0;\ell)$,  $\dircnv{R}(\pi(i),\pi(i+1))$;
\item
A {\em \dwn-path} (to be read as ``down path'') from $w$, of length $\ell \geq 1$, is an undirected path $\pi$ from $w$ of length $\ell$ such that  $\cnv{R}(\pi(\ell-1),\pi(\ell))$;
\item
A {\em \plm-path} (to be read as ``plus-minus path'') from $w$, of length $\ell \geq 2$, is a \dwn-path $\pi$  from $w$ of length $\ell$ such that $R(\pi(0),\pi(1))$;
\item
An {\em \upd-path} (to be read as ``up-down path'') from $w$, of length $2\ell$, for $\ell \geq 1$, 
 is a \plm-path $\pi$  of length $2\ell$ such that $R(\pi(2i),\pi(2i+1))$ and $\cnv{R}(\pi(2i+1),\pi(2i+2))$, for all $i\in [0;\ell)$.\closedefi
\end{itemize}
We call $w$ the {\em starting point} of $\pi$. The {\em ending point} of $\pi$ is $\pi(\ell)$, while
for any $i\in (0;\ell)$, $\pi(i)$ is an {\em intermediate point} of $\pi$.
\closedefi
\end{definition}

Below, we will show some facts regarding the relationship among \upd-paths, \plm-paths and
\dwn-paths, but first we need to introduce some notation and operations on paths over Kripke frames.
For undirected path $\pi$ of length $\ell$ we often use the sequence notation
$(w_i)_{i=0}^{\ell}$ where $w_i=\pi(i)$ for all $i\in [0;\ell]$.

\begin{definition}[Operations on Paths]
Given a Kripke frame $(W,R)$ and  paths
$\pi' = (w'_i)_{i=0}^{\ell'}$ and $\pi'' = (w''_i)_{i=0}^{\ell''}$,
with $w'_{\ell'} = w''_0$, the \emph{sequentialisation} $\pi' \cdot \pi'' : [0;\ell' + \ell''] \to W$
of $\pi'$ with~$\pi''$ is the path 
from~$w'_0$ defined as follows:
$$
(\pi' \cdot \pi'')(i) = 
\left\{
\begin{array}{l}
\pi'(i), \text{ if } i \in [0;\ell'],\\
\pi''(i-\ell'), \text{ if } i \in [\ell'; \ell'+ \ell''].
\end{array}
\right.
$$

For path
$\pi=(w_i)_{i=0}^{\ell}$ and $k \in [0;\ell]$ we define the $k$-shift of $\pi$, denoted by
$\pi{\uparrow} k$, as follows:
$\pi{\uparrow} k = (w_{j+k})_{j=0}^{\ell-k}$ and, for $0<m\leq \ell$, we let
$\pi{\leftarrow}m$ denote the path obtained from $\pi$ by inserting a copy of $\pi(m)$ immediately before $\pi(m)$ itself. In other words, we have: $\pi{\leftarrow}m= (\pi |[0;m])\cdot((\pi(m),\pi(m))\cdot(\pi{\uparrow}m))$.
Finally, any path $\pi |[0;k]$, for some $k\in [0;\ell]$, is a {\em (non-empty) prefix} of $\pi$.
\closedefi
\end{definition}

\begin{figure}[t!]
\begin{center}
\begin{tikzpicture}[scale=0.8, every node/.style={transform shape}]
    \tikzstyle{gstate}=[circle,draw=black,fill=white]
    \tikzset{->-/.style={decoration={
        markings,
        mark=at position #1 with {\arrow{>}}},postaction={decorate}}}
    \node[gstate] (a) at ( 0,0) {a};
    \node[gstate, right of = a] (b) {b};
    \node[gstate, right of = b] (c) {c};
    \node[gstate, right of = c] (d) {d};
    \draw(a) edge[->](b);
    \draw(b) edge[->](c);
    \draw(c) edge[->](d);
\end{tikzpicture}
\end{center}
\caption{A simple finite Kripke frame. Arrows in the figure represent the accessibility relation
$R$.
}\label{fig:FinKriFra}
\end{figure}
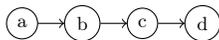

\begin{example}
For Kripke frame $(\SET{a,b,c,d}, R)$ with 
$R=\SET{(a,b),(b,c),(c,d)}$ (see Figure~\ref{fig:FinKriFra}), 
 path $(a,b,c)$ of length $2$ and path $(c,d)$ of length $1$, 
we have that $(a,b,c)\cdot (c,d)=(a,b,c,d)$, of length $3$, $(a)\cdot (a,b) = (a,b)$, $(a)\cdot (a) = (a)$. Note the difference between sequentialisation and concatenation `$++$': for instance, $(a,b){++}(c)=(a,b,c)$ whereas $(a,b)\cdot (c)$  is undefined since $b\not=c$,  $(a){++}(a)$ is $(a,a)$ whereas $(a)\cdot (a) = (a)$.
We have $(a,b,c){\uparrow}1 =(b,c)$ and $(a,b,c){\uparrow}2 =(c)$ 
while $(a,b,c){\leftarrow}1 =(a,b,b,c)$.
Paths $(a), (a,b), (a,b,c)$ are all the (non-empty) prefixes of $(a,b,c)$.
\closeex
\end{example}

As it is clear from Def.~\ref{def:PathsOverKripkeFrames}, every \upd-path is also a \plm-path, that is also a \dwn-path. 
Furthermore, the  three lemmas below ensure that, for reflexive Kripke frames:
\begin{itemize}
\item
for  every \plm-path there is a \upd-path with the same starting and ending points and with the same set of intermediate points, occurring in the same order (Lemma~\ref{lem:pm2ud} below);
\item
for every \dwn-path there is a \upd-path with the same starting and ending points and with the same set of intermediate points, occurring in the same order (Lemma~\ref{lem:d2ud} below);  
\item
for every \dwn-path there is a \plm-path with the same starting and ending points and with the same set of intermediate points, occurring in the same order (Lemma~\ref{lem:d2plm}, proved in Sect.~\ref{apx:prf:lem:d2plm} on page~\pageref{apx:prf:lem:d2plm}).
\end{itemize}

\begin{lemma}\label{lem:pm2ud}
Given a reflexive Kripke frame $(W,R)$ and  a \plm-path $\pi:[0;\ell]\to W$,
there is a \upd-path $\pi':[0;\ell']\to W$, for some $\ell'$, and a 
total, surjective, monotonic non-decreasing function $f:[0;\ell'] \to [0;\ell]$ such that 
$\pi'(j)=\pi(f(j))$ for all $j\in [0;\ell']$.
\end{lemma}

\begin{proof}
We proceed by induction on the length $\ell$ of \plm-path $\pi$.\\
{\bf Base case:} $\ell =2$.\\
In this case, by definition of \plm-path, we have $R(\pi(0),\pi(1))$ and $\cnv{R}(\pi(1),\pi(2))$, which, by definition of \upd-path, implies that $\pi$ itself is an \upd-path and $f:[0;\ell] \to [0;\ell]$ is just the identity function.\\

\noindent
{\bf Induction step.} We assume the assert holds for all \plm-paths of length $\ell$ and we prove it for $\ell+1$.
Let $\pi:[0;\ell+1] \to W$ be a \plm-path. 
Then $\cnv{R}(\pi(\ell),\pi(\ell+1))$, since $\pi$ is a $\pm$-path.
We consider the following cases:\\
{\bf Case A:} $\cnv{R}(\pi(\ell-1),\pi(\ell))$ and  $\cnv{R}(\pi(\ell),\pi(\ell+1))$.\\
In this case, consider the prefix $\pi_1 = \pi | [0;\ell]$ of $\pi$, noting that $\pi_1$ is a \plm-path of length $\ell$. By the Induction Hypothesis there is an \upd-path $\pi'_1$ of some length $\ell'_1$
and a total, surjective, monotonic non-decreasing function $g:[0;\ell'_1] \to [0;\ell]$ such that 
$\pi'_1(j)=\pi_1(g(j))=\pi(g(j))$ for all $j\in [0;\ell'_1]$. 
Note that $\pi'_1(\ell'_1) = \pi(\ell)$ so that  the sequentialisation
of $\pi'_1$ with the two-element path  $(\pi(\ell),\pi(\ell+1))$ is well-defined.
Consider  path $\pi'=(\pi'_1 \cdot (\pi(\ell),\pi(\ell+1)))\leftarrow \ell'_1$, of length $\ell'_1+2$
consisting of $\pi'_1$ followed by $\pi(\ell)$ followed in turn by $\pi(\ell+1)$. 
In other words, $\pi' = (\pi'_1(0) \ldots \pi'_1(\ell'_1),\pi(\ell),\pi(\ell+1))$, with $\pi'_1(\ell'_1) = \pi(\ell)$ --- recall that $R$ is reflexive.
It is easy to see that $\pi'$ is an \upd-path and that function $f:[0;\ell'_1+2] \to [0;\ell+1]$, with
$f(j)=g(j)$ for $j\in [0;\ell'_1]$, $f(\ell'_1+1)=\ell$ and $f(\ell'_1+2)=\ell+1$, is 
total, surjective, and monotonic non-decreasing.\\
{\bf Case B:} $R(\pi(\ell-1),\pi(\ell))$ and $\cnv{R}(\pi(\ell),\pi(\ell+1))$.\\
In this case the prefix $\pi | [0;\ell]$ of $\pi$ is {\em not} a \plm-path.
We then consider the path consisting of prefix $\pi|[0;\ell-1]$ where we add a copy
of $\pi(\ell-1)$, i.e. the path $\pi_1=(\pi|[0;\ell-1])\leftarrow (\ell-1)$ --- we can do that because $R$ is reflexive.
Note that $\pi_1$ is a \plm-path and has length $\ell$.
By the Induction Hypothesis there is an \upd-path $\pi'_1$ of some length $\ell'_1$
and a total, surjective, monotonic non-decreasing function $g:[0;\ell'_1] \to [0;\ell]$ such that 
$\pi'_1(j)=\pi_1(g(j))=\pi(g(j))$ for all $j\in [0;\ell'_1]$. 
Consider path $\pi'=\pi'_1 \cdot (\pi(\ell-1),\pi(\ell), \pi(\ell+1))$, of length $\ell'_1 +2$,
that is well defined since $\pi'_1(\ell'_1)=\pi(\ell-1)$ by definition of $\pi_1$.
In other words, $\pi'=(\pi'_1(0),\ldots,\pi'_1(\ell'_1),\pi(\ell),\pi(\ell+1))$, with $\pi'_1(\ell'_1) = \pi(\ell-1)$.
Path $\pi'$ is an \upd-path. In fact $\pi'|[0;\ell'_1]=\pi'_1$ is an \upd-path. Furthermore, 
$\pi'(\ell'_1)=\pi(\ell-1)$, $R(\pi(\ell-1),\pi(\ell))$, $\cnv{R}(\pi(\ell)),\pi(\ell+1)$ and $\pi(\ell+1)=\pi'(\ell'_1 +2)$.
Finally, function $f:[0;\ell'_1+2] \to [0;\ell+1]$, with
$f(j)=g(j)$ for $j\in [0;\ell'_1]$, $f(\ell'_1+1)=\ell$ and $f(\ell'_1+2)=\ell+1$, is 
total, surjective, and monotonic non-decreasing.
\end{proof}

\begin{lemma}\label{lem:d2ud}
Given a reflexive Kripke frame  $(W,R)$ and  a \dwn-path $\pi:[0;\ell]\to W$,
there is a \upd-path $\pi':[0;\ell'']\to W$, for some $\ell'$, and a 
total, surjective, monotonic non-decreasing function $f:[0;\ell'] \to [0;\ell]$ such that 
$\pi'(j)=\pi(f(j))$ for all $j\in [0;\ell']$.
\end{lemma}

\begin{proof}
The proof is carried out by induction on the length $\ell$ of $\pi$.\\
{\bf Base case.} $\ell=1$.
Suppose $\ell=1$, i.e. $\pi:[0;1] \to W$ with $\cnv{R}(\pi(0),\pi(1))$. Then let 
$\pi':[0;2]\to W$ be such that $\pi'(0)=\pi'(1)=\pi(0)$ and $\pi'(2)=\pi(1)$ --- we can do that since $R$ is reflexive --- and
$f:[0;2] \to [0;1]$ be such that $f(0)=f(1)=0$ and $f(2)=1$.
Clearly $\pi'$ is an \upd-path and $\pi'(j)=\pi(f(j))$ for all $j\in [0;2]$.\\
{\bf Induction step.} We assume the assert holds for all \dwn-paths of length $\ell$ and we prove it for $\ell+1$.
Let $\pi:[0;\ell+1] \to W$ a \dwn-path and suppose the assert holds for all \dwn-paths of length $\ell$. 
In particular, it holds for $\pi\uparrow 1$, i.e., there is a \upd-path $\pi''$ 
of some length $\ell''$ with $\pi''(0) = \pi(1)$, and 
total, monotonic non-decreasing surjection $g:[0;\ell'']\to W$ such that 
$\pi''(j) = \pi(g(j))$ for all $j\in [0;\ell'']$.
Suppose $R(\pi(0),\pi(1))$ does not hold. Then, since $R$ is reflexive, we let $\pi'= (\pi(0),\pi(0),\pi(1))\cdot \pi''$ and
$f :[0;\ell''+2] \to [0;\ell+1]$ with $f(0)=f(1)=0$ and
$f(j)=g(j-2)$ for all $j\in[2;\ell''+2]$. 
If instead $R(\pi(0),\pi(1))$, then we let $\pi'= (\pi(0),\pi(1),\pi(1))\cdot \pi''$ and
$f :[0;\ell''+2] \to [0;\ell+1]$ with $f(0)=0, f(1)=1$ and
$f(j)=g(j-2)$ for all $j\in[2;\ell''+2]$.
\end{proof}

The following result states that to evaluate an \slcsE{}  formula $\eta(\form_1,\form_2)$ in a poset model, it does not matter whether one considers \plm-paths or \dwn-paths. 

\begin{proposition}\label{prop:interchangeableE}
Given a finite poset $\calF=(W,\preceq,\peval)$, $w \in W$ and an \slcsE{} formula $\eta(\form_1,\form_2)$ the following statements are equivalent:
\begin{enumerate}
\item\label{enu:plmE}
There exists a \plm-path $\pi:[0;\ell] \to W$ for some $\ell$ with
$\pi(0)=w$, 
$\calF, \pi(\ell)\models \form_2$ and 
$\calF, \pi(i)\models \form_1$  for all $i\in (0;\ell)$.
\item\label{enu:dwnE}
There exists a \dwn-path $\pi:[0;\ell] \to W$ for some $\ell$ with
$\pi(0)=w$, 
$\calF, \pi(2\ell)\models \form_2$ and 
$\calF, \pi(i)\models \form_1$  for all $i\in (0;\ell)$.
\end{enumerate}
\end{proposition}
\begin{proof}
The equivalence of statements (\ref{enu:plmE}) and (\ref{enu:dwnE}) follows directly 
from Lemma~\ref{lem:d2plm} and the fact that \plm-paths are also \dwn-paths.
\end{proof}

\begin{lemma}\label{lem:dPExists4udP}
Given a finite poset model $\calF=(W,\preceq, \peval)$ and  a
weak \plm-bisimulation $B \subseteq W \times W$, for all $w_1,w_2$ such that $B(w_1,w_2)$, the following holds:
for each  \upd-path $\pi_1:[0;2h] \to W$ from $w_1$ 
there is a \dwn-path $\pi_2:[0;k] \to W$ from $w_2$ 
such that $B(\pi_1(2h),\pi_2(k))$ and 
for each $j\in [0;k)$ there is $i\in [0;2h)$ such that $B(\pi_1(i),\pi_2(j))$.
\end{lemma}

\begin{proof}
We prove the assert by induction on $h$.\\
{\bf Base case.} $h=1$.\\
If $h=1$, the assert follows directly from Definition~\ref{def:WPLMBis} on page~\pageref{def:WPLMBis}
where $w_1= \pi(0), u_1=\pi(1)$ and $d_1=\pi(2)$.\\
{\bf Induction step.} We assume the assert holds for \upd-paths of length $2h$ or less and we prove it for \upd-paths of length $2(h+1)$. \\
Suppose $\pi_1$ is a \upd-path of length $2h+2$ and consider \upd-path $\pi'_1=\pi_1|[0;2h]$.
By the Induction Hypothesis, we know that there is a \dwn-path $\pi'_2:[0;k']\to W$ from $w_2$
such that $B(\pi'_1(2h),\pi'_2(k'))$ and for each $j\in[0;k')$ there is $i\in [0;2h)$ such that
$B(\pi'_1(i),\pi'_2(j))$. Clearly, this means that $B(\pi_1(2h),\pi'_2(k'))$ and for each $j\in[0;k')$ there is $i\in [0;2h)$ such that $B(\pi_1(i),\pi'_2(j))$.
Furthermore, since $B(\pi_1(2h),\pi'_2(k'))$ and $B$ is a weak \plm-bisimulation, 
we also know that there is a \dwn-path 
$\pi''_2:[0;k'']\to W$ from $\pi'_2(k')$ such that $B(\pi_1(2h+2),\pi''_2(k''))$ and for each $j\in[0;k'')$
there is $i\in [2h;2h+2)$ such that $B(\pi_1(i),\pi'_2(j))$. Let $\pi_2:[0;k'+k'']\to W$ be defined as
$\pi_2= \pi'_2 \cdot \pi''_2$. Clearly $\pi_2$ is a \dwn-path, since so is $\pi''_2$.
Furthermore $B(\pi_1(2h+2),\pi_2(k'+k''))$ since $B(\pi_1(2h+2),\pi''_2(k''))$ and $\pi''_2(k'')=\pi_2(k'+k'')$. Finally, it is straightforward to check  for all $j\in [0;k'+k'')$ 
there is $i\in [0;2h+2)$ such that $B(\pi_1(i),\pi_2(j))$.
\end{proof}

\section{Detailed Proofs}\label{apx:DetailedProofs}

\subsection{Proof of Lemma~\ref{lem:d2plm}}\label{apx:prf:lem:d2plm}
\noindent
{\bf Lemma~\ref{lem:d2plm}.}
{\em 
Given a reflexive Kripke frame  $(W,R)$ and a \dwn-path $\pi:[0;\ell]\to W$,
there is a \plm-path $\pi':[0;\ell'']\to W$, for some $\ell'$, and a 
total, surjective, monotonic, non-decreasing function $f:[0;\ell'] \to [0;\ell]$ with  
$\pi'(j)=\pi(f(j))$ for all $j\in [0;\ell']$.
}

\begin{proof}
The assert follows directly from Lemma~\ref{lem:d2ud} on page~\pageref{lem:d2ud} since every \upd-path is also a \plm-path.
\end{proof}

\subsection{Proof of Lemma~\ref{lem:etgaCorrectG}}\label{apx:Prf:lem:etgaCorrectG}
\noindent
{\bf Lemma~\ref{lem:etgaCorrectG}.}
{\em
Let $\calP=(|K|,V)$ be a polyhedral model, $x \in |K|$ and $\form$ a \slcsE{} formula.
Then $\calP,x \models \form$ iff $\calP,x \models \etga(\form)$.
}

\begin{proof}
By induction on the structure of $\form$. We consider only the case $\eta(\form_1,\form_2)$.
Suppose $\calP,x \models \eta(\form_1,\form_2)$. By definition there is a topological path 
$\pi$ such that $\calP,\pi(1) \models \form_2$ and 
$\calP,\pi(r) \models \form_1$ for all $r \in [0,1)$. By the Induction Hypothesis this is the same to say that $\calP,\pi(1) \models \etga(\form_2)$ and 
$\calP,\pi(r) \models \etga(\form_1)$ for all $r \in [0,1)$, i.e.
$\calP,x \models \etga(\form_1)$, 
$\calP,\pi(1) \models \etga(\form_2)$ and 
$\calP,\pi(r) \models \etga(\form_1)$ for all $r \in (0,1)$. In other words, we have
$\calP,x \models \etga(\form_1) \land \gamma(\etga(\form_1), \etga(\form_2))$ that,
by Definition~\ref{def:etga} on page~\pageref{def:etga} means $\calP,x \models \etga(\eta(\form_1,\form_2))$.

Suppose now $\calP,x \models \etga(\eta(\form_1,\form_2))$, i.e.
$\calP,x \models \etga(\form_1) \land \gamma(\etga(\form_1), \etga(\form_2))$, by Definition~\ref{def:etga} on page~\pageref{def:etga}.
Since $\calP,x \models \gamma(\etga(\form_1), \etga(\form_2))$, there is a path $\pi$
such that $\calP,\pi(1) \models \etga(\form_2)$ and 
$\calP,\pi(r) \models \etga(\form_1)$ for all $r \in (0,1)$. Using the Induction Hypothesis
we know the following holds:
$\calP,x \models  \form_1$, $\calP,\pi(1) \models \form_2$, and
$\calP,\pi(r) \models \form_1$ for all $r \in (0,1)$, i.e.
$\calP,\pi(1) \models \form_2$ and $\calP,\pi(r) \models \form_1$ for all $r \in [0,1)$.
So, we get $\calP,x \models \eta(\form_1,\form_2)$.
\end{proof}

\subsection{Proof concerning the example of Remark~\ref{rem:weakGweakerG}}
\label{apx:prf:rem:weakGweakerG}
The assert can be proven by induction on the structure of formulas.
The case for proposition letters, negation  and conjunction are straightforward and omitted.

Suppose $\calP, A \models \eta(\form_1, \form_2)$. 
Then there is a topological path $\pi_A:[0,1] \to |K|$ from $A$ such that
$\calP, \pi_A(1) \models \form_2$ and $\calP, \pi_A(r)\models \form_1$ for all 
$r \in [0,1)$. 
Since $\calP, A \models \form_1$, by the Induction Hypothesis, we have that $\calP, x \models \form_1$ for all $x \in \relint{ABC}$.
For each $x\in\relint{ABC}$, define $\pi_x: [0,1] \to |K|$ as follows, for arbitrary $v \in (0,1)$:
$$
\pi_x(r)=
\left\{
\begin{array}{l}
\frac{r}{v}A + \frac{v-r}{v}x, \mbox{ if } r\in [0,v),\\\\
\pi_A(\frac{r-v}{1-v}), \mbox{ if }r\in [v,1].
\end{array}
\right.
$$
Function $\pi_x$ is continuous. Furthermore, for all $y\in [0,v)$, we have that $\calP, \pi_x(y) \models \form_1$,  since $\pi_x(y) \in \relint{ABC}$. Also, for all $y\in [v,1)$ we have that $\calP, \pi_x(y) \models \form_1$,  since $\pi_x(y) = \pi_A(\frac{y-v}{1-v})$,
$0\leq \frac{y-v}{1-v}<1$ and for $y \in [0,1)$ we have that $\calP,\pi_A(y), \models \form_1$.  Thus
$\calP,\pi_x(r) \models \form_1$ for all $r \in [0,1)$. Finally, $\pi_x(1)=\pi_A(1)$ and 
$\calP, \pi_A(1) \models \form_2$ by hypothesis. 
Thus, $\pi_x$ is a topological path that witnesses $\calP, x \models \eta(\form_1, \form_2)$.

The proof of the converse is similar, using instead function $\pi_A: [0,1] \to |K|$ defined as follows, for arbitrary $v \in (0,1)$:
$$
\pi_A(r)=
\left\{
\begin{array}{l}
\frac{r}{v}p + \frac{v-r}{v}A, \mbox{ if } r\in [0,v),\\\\
\pi_p(\frac{r-v}{1-v}), \mbox{ if }r\in [v,1].
\end{array}
\right.
$$

\subsection{Proof of Proposition~\ref{prop:chiE}}\label{apx:prf:prop:chiE}
\noindent
{\bf Proposition~\ref{prop:chiE}.}
{\em
Given a finite poset model $(W,\preceq,\peval)$, for $w_1,w_2 \in W$, it holds that 
$$
  \calF,w_2 \models \chi(w_1) \mbox{ if and only if } w_1 \slcsEeq w_2.
$$
}
\begin{proof}
Suppose $w_1 \not\slcsEeq w_2$, then we have $\calF,w_2 \not\models \delta_{w_1,w_2}$, and so
  $\calF,w_2 \not\models \bigwedge_{w \in W} \: \delta_{w_1,w}$.
  If, instead, $w_1 \slcsEeq w_2$, then we have: $\delta_{w_1,w_1} \equiv \delta_{w_1,w_2} \equiv \ltrue  $ by definition, since $w_1 \slcsEeq w_1$ and $w_1 \slcsEeq w_2$. Moreover, for any other $w$, we have that, in any case,
  $\calF,w_1\models \delta_{w_1,w}$ holds and since $w_1 \slcsEeq w_2$, also $\calF,w_2\models \delta_{w_1,w}$ holds.
  So, in conclusion, $\calF,w_2\models \bigwedge_{w \in W} \: \delta_{w_1,w}$.
\end{proof}

\subsection{Proof of Lemma~\ref{lem:etgaCorrectE}}\label{apx:prf:lem:etgaCorrectE} 
\noindent
{\bf Lemma~\ref{lem:etgaCorrectE}.}
{\em
Let $\calF=(W,\preceq,\peval)$ be a 
finite poset model, 
$w \in W$ and $\form$ a \slcsE{} formula.
Then $\calF,w \models \form$ iff $\calF,w \models \etga(\form)$.
}

\begin{proof}
Similar to that of Lemma~\ref{lem:etgaCorrectG}, but with reference to the finite poset intepretation of the logic.
\end{proof}

\subsection{Proof concerning the example of Remark~\ref{rem:weakPMeakerPM}}\label{apx:prf:rem:weakPMeakerPM}

We prove the assert by induction on the structure of formulas. The case for atomic proposition letters, negation  and conjunction are straightforward and omitted.
Suppose $\calF,\relint{A} \models \eta(\form_1, \form_2)$. Then, there is a \plm-path $\pi$ of 
some length $\ell\geq2$ such that
$\pi(0) = \relint{A}$, $\pi(\ell) \models \form_2$ and $\pi(i)\models \form_1$ for all 
$i \in [0;\ell)$. Since $\calF,\relint{A} \models \form_1$, by the Induction Hypothesis, we have that $\calF,\relint{ABC} \models \form_1$.
Consider then path $\pi'= (\relint{ABC},\relint{ABC},\relint{A})\cdot \pi$. Path $\pi'$ is a \plm-path and it witnesses $\calF, \relint{ABC} \models \eta(\form_1, \form_2)$.

Suppose now $\calF, \relint{ABC} \models \eta(\form_1, \form_2)$ and let $\pi$ be a \plm-path witnessing it.
Then, path $(\relint{A},\relint{ABC},\relint{ABC}) \cdot \pi$ is a \plm-path witnessing $\calF, \relint{A} \models \eta(\form_1, \form_2)$.

\subsection{Proof of Lemma~\ref{lem:VB}}\label{apx:prf:lem:VB}
\noindent
{\bf Lemma~\ref{lem:VB}.}
{\em 
Given a polyhedral model $\calP=(|K|,V)$,
for all $x\in |K|$ and formulas $\form$ of \slcsE{}
the following holds: $\calP,x \models \form$ if and only if $\map(\calP),\map(x) \models \etga(\form)$.
}

\begin{proof}
The proof is by induction on the structure of $\form$. 
We consider only the case $\eta(\form_1,\form_2)$.
Suppose $\calP,x \models \eta(\form_1,\form_2)$. 
By Lemma~\ref{lem:etgaCorrectG} on page~\pageref{lem:etgaCorrectG} we get 
$\calP,x \models \etga(\eta(\form_1,\form_2))$
and then by Definition~\ref{def:etga} on page~\pageref{def:etga}, we have $\calP,x \models \etga(\form_1) \land \gamma(\etga(\form_1),\etga(\form_2))$, 
that is $\calP,x \models \etga(\form_1)$ and  $\calP,x \models\gamma(\etga(\form_1),\etga(\form_2))$. Again by Lemma~\ref{lem:etgaCorrectG} on page~\pageref{lem:etgaCorrectG}, we get also 
$\calP,x \models \form_1$ and so, by the Induction Hypothesis, we have
$\map(\calP),\map(x) \models \etga(\form_1)$.
Furthermore, by Theorem 4.4 of~\cite{Be+22} we also get 
$\map(\calP),\map(x) \models \gamma(\etga(\form_1),\etga(\form_2))$.
Thus we get $\map(\calP),\map(x) \models \etga(\form_1)\land \gamma(\etga(\form_1),\etga(\form_2))$, that is
$\map(\calP),\map(x) \models \etga(\eta(\form_1,\form_2))$.\\
Suppose now $\map(\calP),\map(x) \models \etga(\eta(\form_1,\form_2))$.
This means $\map(\calP),\map(x) \models \etga(\form_1) \land \gamma(\etga(\form_1),\etga(\form_2))$,
that is $\map(\calP),\map(x) \models \etga(\form_1)$  and 
$\map(\calP),\map(x) \models \gamma(\etga(\form_1),\etga(\form_2))$.
By the Induction Hypothesis we get that $\calP,x \models \form_1$.
Furthermore, by Theorem 4.4 of~\cite{Be+22} we also get
$\calP,x \models \gamma(\etga(\form_1),\etga(\form_2))$.
This means that there is topological path $\pi$ such that 
$\calP,\pi(1) \models \etga(\form_2)$ and $\calP,\pi(r) \models \etga(\form_1)$
for all $r \in (0,1)$. 
Using Lemma~\ref{lem:etgaCorrectG} on page~\pageref{lem:etgaCorrectG} we also get
$\calP,\pi(1) \models \form_2$ and $\calP,\pi(r) \models \form_1$
for all $r \in (0,1)$ and since also $\calP,x \models \form_1$ (see above), we
get $\calP,\pi(1) \models \form_2$ and $\calP,\pi(r) \models \form_1$
for all $r \in [0,1)$, that is
$\calP,x \models \eta (\form_1,\form_2)$.
\end{proof}

\subsection{Proof of Theorem~\ref{theo:calMPresForm}}\label{apx:prf:theo:calMPresForm}
\noindent
{\bf Theorem~\ref{theo:calMPresForm}.}
{\em 
Given a polyhedral model $\calP=(|K|,V)$,
for all $x\in |K|$ and formulas $\form$ of \slcsE{}
it holds that: $\calP,x \models \form$ if and only if $\map(\calP),\map(x) \models \form$.
}

\begin{proof}
Using Lemma~\ref{lem:VB} on page~\pageref{lem:VB}, we know that 
$\calP,x \models \form$ if and only if $\map(\calP),\map(x) \models \etga(\form)$.
Moreover, by Lemma~\ref{lem:etgaCorrectE} on page~\pageref{lem:etgaCorrectE}, we know that 
$\map(\calP),\map(x) \models \etga(\form)$ if and only if
$\map(\calP),\map(x) \models \form$, which brings to the result.
\end{proof}

\subsection{Proof of Lemma~\ref{lem:dPExists}}\label{apx:prf:lem:dPExists}
{\bf Lemma~\ref{lem:dPExists}.}
{\em Given a finite poset model $\calF=(W,\preceq, \peval)$ and  
weak \plm-bisimulation $B \subseteq W \times W$, for all $w_1,w_2$ such that $B(w_1,w_2)$, the following holds:
for each  \dwn-path $\pi_1:[0;k_1] \to W$ from $w_1$ 
there is a \dwn-path $\pi_2:[0;k_2] \to W$ from $w_2$ 
such that $B(\pi_1(k_1),\pi_2(k_2))$ and 
for each $j\in [0;k_2)$ there is $i\in [0;k_1)$ such that $B(\pi_1(i),\pi_2(j))$.
}

\begin{proof}
Let $\pi_1:[0;k_1] \to W$ be a \dwn-path from $w_1$. By Lemma~\ref{lem:d2ud} on page~\pageref{lem:d2ud} we know that
there is an \upd-path $\hat{\pi}_1:[0;2h] \to W$ and total, monotonic non-decreasing surjection
$f:[0;2h] \to [0;k_1]$ such that $\hat{\pi}_1(j)=\pi_1(f(j))$  for all $j\in [0;2h]$. Furthermore, by Lemma~\ref{lem:dPExists4udP} on page~\pageref{lem:dPExists4udP}, we know that
there is a \dwn-path $\pi_2:[0;k_2] \to W$ from $w_2$ 
such that $B(\hat{\pi}_1(2h),\pi_2(k_2))$ and 
for each $j\in [0;k_2)$ there is $i\in [0;2h)$ such that $B(\hat{\pi}_1(i),\pi_2(j))$.
In addition, 
$\hat{\pi}_1(0)=\pi_1(0)=w_1$,  $B(\pi_1(k_1),\pi_2(k_2))$ since $B(\hat{\pi}_1(2h),\pi_2(k_2))$ and $\hat{\pi}_1(2h)=\pi_1(k_1)$. Finally, 
for each $j\in [0;k_2)$ there is $i \in [0;k_1)$ such that $B(\pi_1(i),\pi_2(j))$,
since there is $n\in [0;2h)$ such that $B(\hat{\pi}_1(n),\pi_2(j))$ and $f(n)=i$ for some $i \in [0;k_1)$.
\end{proof}

\subsection{Proof of Lemma~\ref{lem:dPtoTP}}\label{apx:prf:lem:dPtoTP}

\noindent
{\bf Lemma~\ref{lem:dPtoTP}.}
{\em
Given a polyhedral model $\calP=(|K|,V)$, and associated cell poset model $\map(\calP)=(W,\preceq,\peval)$, for any  
\dwn-path  $\pi:[0;\ell] \to W$, 
there is a topological path $\pi':[0,1] \to |K|$ such that: (i) $\map(\pi'(0))=\pi(0)$, (ii) $\map(\pi'(1))=\pi(\ell)$, and
(iii) for all $r \in (0,1)$ there is $i<\ell$ such that $\map(\pi'(r))=\pi(i)$.
}

\begin{proof}
Since $\pi$ is a \dwn-path, we have that either $\closure_T(\map^{-1}(\pi(k-1))) \sqsubseteq \closure_T(\map^{-1}(\pi(k)))$  
or $\closure_T(\map^{-1}(\pi(k))) \sqsubseteq \closure_T(\map^{-1}(\pi(k-1)))$,  
for each $k\in (0;\ell]$\footnote{We recall here that $\sigma_1 \sqsubseteq \sigma_2$ iff 
$\relint{\sigma_1} \preceq \relint{\sigma_2}$ and that $\sigma = \closure_T(\relint{\sigma})$.}. 
It follows that there is a continuous map 
$\pi'_k:[\frac{k-1}{\ell},\frac{k}{\ell}]\to |K|$ such that, in the first case, 
$\map(\pi'_k(\frac{k-1}{\ell})) = \pi(k-1)$ and
$\pi'_k((\frac{k-1}{\ell},\frac{k}{\ell}]) \subseteq \closure_T(\map^{-1}(\pi(k)))$, 
while in the second case, 
$\pi'_k([\frac{k-1}{\ell},\frac{k}{\ell})) \subseteq \closure_T(\map^{-1}(\pi(k-1)))$ and 
$\map(\pi'_k(\frac{k}{\ell}))=\pi(k)$.
In fact $\pi'_k$ can be realised as a linear bijection to the line segment connecting the barycenters  
in the corresponding cell, either in $\map^{-1}(\pi(k))$ or in $\map^{-1}(\pi(k-1))$, respectively.

For each $k\in(0;\ell)$, both $\pi'_k(\frac{k}{\ell})$ and $\pi'_{k+1}(\frac{k}{\ell})$ coincide with the barycenter of $\map^{-1}(\pi(k))$, so that defining $\pi'(r)=\pi'_k(r)$ for $r\in[\frac{k-1}{\ell},\frac{k}{\ell}]$ correctly defines a  topological path (actually a piece-wise linear path), satisfying (i) and (ii).
Finally since $\pi$ is a \dwn-path, $\pi(\ell) \preceq\pi(\ell-1)$, so that 
$\pi'([\frac{\ell-1}{\ell},1))\subseteq\map^{-1}(\pi(\ell-1))$. This implies (iii) above.
\end{proof}

%% file: BCGJLMdV24.bbl
\begin{thebibliography}{10}
\providecommand{\url}[1]{\texttt{#1}}
\providecommand{\urlprefix}{URL }
\providecommand{\doi}[1]{https://doi.org/#1}

\bibitem{Ba+20}
{Banci Buonamici}, F., Belmonte, G., Ciancia, V., Latella, D., Massink, M.:
  Spatial logics and model checking for medical imaging. Int. J. Softw. Tools
  Technol. Transf.  \textbf{22}(2),  195--217 (2020),
  \url{https://doi.org/10.1007/s10009-019-00511-9}

\bibitem{Be+21}
Belmonte, G., Broccia, G., Ciancia, V., Latella, D., Massink, M.: Feasibility
  of spatial model checking for nevus segmentation. In: Bliudze, S., Gnesi, S.,
  Plat, N., Semini, L. (eds.) 9th {IEEE/ACM} International Conference on Formal
  Methods in Software Engineering, FormaliSE@ICSE 2021, Madrid, Spain, May
  17-21, 2021. pp. 1--12. {IEEE} (2021),
  \url{https://doi.org/10.1109/FormaliSE52586.2021.00007}

\bibitem{Be+19}
Belmonte, G., Ciancia, V., Latella, D., Massink, M.: Voxlogica: {A} spatial
  model checker for declarative image analysis. In: Vojnar, T., Zhang, L.
  (eds.) Tools and Algorithms for the Construction and Analysis of Systems -
  25th International Conference, {TACAS} 2019, Held as Part of the European
  Joint Conferences on Theory and Practice of Software, {ETAPS} 2019, Prague,
  Czech Republic, April 6-11, 2019, Proceedings, Part {I}. Lecture Notes in
  Computer Science, vol. 11427, pp. 281--298. Springer (2019),
  \url{https://doi.org/10.1007/978-3-030-17462-0\_16}

\bibitem{vBB07}
van Benthem, J., Bezhanishvili, G.: Modal logics of space. In: Aiello, M.,
  Pratt{-}Hartmann, I., Benthem, J.v. (eds.) Handbook of Spatial Logics, pp.
  217--298. Springer (2007), \url{https://doi.org/10.1007/978-1-4020-5587-4\_5}

\bibitem{Be+22}
Bezhanishvili, N., Ciancia, V., Gabelaia, D., Grilletti, G., Latella, D.,
  Massink, M.: {Geometric Model Checking of Continuous Space}. {Log. Methods
  Comput. Sci.}  \textbf{18}(4),  7:1--7:38 (2022),
  \url{https://lmcs.episciences.org/10348}, {DOI 10.46298/LMCS-18(4:7)2022.
  Published on line: Nov 22, 2022. ISSN: 1860-5974}

\bibitem{BMMP2018}
Bezhanishvili, N., Marra, V., McNeill, D., Pedrini, A.: Tarski's theorem on
  intuitionistic logic, for polyhedra. Annals of Pure and Applied Logic
  \textbf{169}(5),  373--391 (2018).
  \doi{https://doi.org/10.1016/j.apal.2017.12.005},
  \url{https://www.sciencedirect.com/science/article/pii/S016800721730146X}

\bibitem{Ci+22}
Ciancia, V., Latella, D., Massink, M., {de Vink}, E.P.: Back-and-forth in
  space: On logics and bisimilarity in closure spaces. In: Jansen, N.,
  Stoelinga, M., , {van den Bos}, P. (eds.) {A Journey From Process Algebra via
  Timed Automata to Model Learning - A Festschrift Dedicated to Frits
  Vaandrager on the Occasion of His 60th Birthday}. Lecture Notes in Computer
  Science, vol. 13560, pp. 98--115. Springer (2022)

\bibitem{Ci+23c}
Ciancia, V., Gabelaia, D., Latella, D., Massink, M., de~Vink, E.P.: On
  bisimilarity for polyhedral models and {SLCS}. In: Huisman, M., Ravara, A.
  (eds.) Formal Techniques for Distributed Objects, Components, and Systems -
  43rd {IFIP} {WG} 6.1 International Conference, {FORTE} 2023, Held as Part of
  the 18th International Federated Conference on Distributed Computing
  Techniques, DisCoTec 2023, Lisbon, Portugal, June 19-23, 2023, Proceedings.
  Lecture Notes in Computer Science, vol. 13910, pp. 132--151. Springer (2023).
  \doi{10.1007/978-3-031-35355-0\_9},
  \url{https://doi.org/10.1007/978-3-031-35355-0\_9}

\bibitem{Ci+18}
Ciancia, V., Gilmore, S., Grilletti, G., Latella, D., Loreti, M., Massink, M.:
  Spatio-temporal model checking of vehicular movement in public transport
  systems. Int. J. Softw. Tools Technol. Transf.  \textbf{20}(3),  289--311
  (2018), \url{https://doi.org/10.1007/s10009-018-0483-8}

\bibitem{Ci+15}
Ciancia, V., Grilletti, G., Latella, D., Loreti, M., Massink, M.: An
  experimental spatio-temporal model checker. In: Bianculli, D., Calinescu, R.,
  Rumpe, B. (eds.) Software Engineering and Formal Methods - {SEFM} 2015
  Collocated Workshops: ATSE, HOFM, MoKMaSD, and VERY*SCART, York, UK,
  September 7-8, 2015, Revised Selected Papers. Lecture Notes in Computer
  Science, vol.~9509, pp. 297--311. Springer (2015),
  \url{https://doi.org/10.1007/978-3-662-49224-6\_24}

\bibitem{Ci+23a}
Ciancia, V., Groote, J., Latella, D., Massink, M., {de Vink}, E.: Minimisation
  of spatial models using branching bisimilarity. In: Chechik, M., Katoen,
  J.P., Leucker, M. (eds.) 25th International Symposium, FM 2023, L\"{u}beck,
  March 6–10, 2023, Proceedings. Lecture Notes in Computer Science, vol.
  14000, p. 263–281. Springer (2023). \doi{10.1007/978-3-031-27481-7_16}

\bibitem{Ci+14}
Ciancia, V., Latella, D., Loreti, M., Massink, M.: Specifying and verifying
  properties of space. In: D{\'{\i}}az, J., Lanese, I., Sangiorgi, D. (eds.)
  Theoretical Computer Science - 8th {IFIP} {TC} 1/WG 2.2 International
  Conference, {TCS} 2014, Rome, Italy, September 1-3, 2014. Proceedings.
  Lecture Notes in Computer Science, vol.~8705, pp. 222--235. Springer (2014),
  \url{https://doi.org/10.1007/978-3-662-44602-7\_18}

\bibitem{Ci+16}
Ciancia, V., Latella, D., Loreti, M., Massink, M.: Model checking spatial
  logics for closure spaces. Log. Methods Comput. Sci.  \textbf{12}(4) (2016),
  \url{https://doi.org/10.2168/LMCS-12(4:2)2016}

\bibitem{Ci+19b}
Ciancia, V., Latella, D., Massink, M.: Embedding {RCC8D} in the collective
  spatial logic {CSLCS}. In: Boreale, M., Corradini, F., Loreti, M., Pugliese,
  R. (eds.) Models, Languages, and Tools for Concurrent and Distributed
  Programming - Essays Dedicated to Rocco De Nicola on the Occasion of His 65th
  Birthday. Lecture Notes in Computer Science, vol. 11665, pp. 260--277.
  Springer (2019), \url{https://doi.org/10.1007/978-3-030-21485-2\_15}

\bibitem{CLMP15}
Ciancia, V., Latella, D., Massink, M., Pa\v{s}kauskas, R.: Exploring
  spatio-temporal properties of bike-sharing systems. In: 2015 {IEEE}
  International Conference on Self-Adaptive and Self-Organizing Systems
  Workshops, {SASO} Workshops 2015, Cambridge, MA, USA, September 21-25, 2015.
  pp. 74--79. {IEEE} Computer Society (2015),
  \url{https://doi.org/10.1109/SASOW.2015.17}

\bibitem{Ci+16a}
Ciancia, V., Latella, D., Massink, M., Pa\v{s}kauskas, R., Vandin, A.: A
  tool-chain for statistical spatio-temporal model checking of bike sharing
  systems. In: Margaria, T., Steffen, B. (eds.) Leveraging Applications of
  Formal Methods, Verification and Validation: Foundational Techniques - 7th
  International Symposium, ISoLA 2016, Imperial, Corfu, Greece, October 10-14,
  2016, Proceedings, Part {I}. Lecture Notes in Computer Science, vol.~9952,
  pp. 657--673 (2016), \url{https://doi.org/10.1007/978-3-319-47166-2\_46}

\bibitem{Ci+23}
Ciancia, V., Latella, D., Massink, M., de~Vink, E.P.: On bisimilarity for
  quasi-discrete closure spaces (2023), \url{https://arxiv.org/abs/2301.11634}

\bibitem{Gr+17}
Groote, J.F., Jansen, D.N., Keiren, J.J.A., Wijs, A.: An
  \emph{O}(\emph{m}log\emph{n}) algorithm for computing stuttering equivalence
  and branching bisimulation. {ACM} Trans. Comput. Log.  \textbf{18}(2),
  13:1--13:34 (2017), \url{https://doi.org/10.1145/3060140}

\bibitem{LevinePRZ2012}
Levine, J.A., Paulsen, R.R., Zhang, Y.: Mesh processing in medical-image
  analysis -- a tutorial. IEEE Computer Graphics and Applications
  \textbf{32}(5),  22--28 (2012). \doi{10.1109/MCG.2012.91}

\bibitem{LoQ23}
Loreti, M., Quadrini, M.: A spatial logic for simplicial models. Log. Methods
  Comput. Sci.  \textbf{19}(3) (2023). \doi{10.46298/LMCS-19(3:8)2023},
  \url{https://doi.org/10.46298/lmcs-19(3:8)2023}

\bibitem{McKT44}
McKinsey, J., Tarski, A.: The algebra of topology. Annals of Mathematics
  \textbf{45},  141--191 (1944). \doi{10.2307/1969080}

\bibitem{Ne+18}
Nenzi, L., Bortolussi, L., Ciancia, V., Loreti, M., Massink, M.: Qualitative
  and quantitative monitoring of spatio-temporal properties with {SSTL}. Log.
  Methods Comput. Sci.  \textbf{14}(4) (2018),
  \url{https://doi.org/10.23638/LMCS-14(4:2)2018}

\end{thebibliography}
